\documentclass[pdflatex,sn-mathphys-num]{sn-jnl}

\usepackage{graphicx}%
\usepackage{multirow}%
\usepackage{amsmath,amssymb,amsfonts}%
\usepackage{amsthm}%
\usepackage{mathrsfs}%
\usepackage[title]{appendix}%
\usepackage{xcolor}%
\usepackage{textcomp}%
\usepackage{manyfoot}%
\usepackage{booktabs}%
\usepackage{algorithm}%
\usepackage{algorithmicx}%
\usepackage{algpseudocode}%
\usepackage{listings}%

\usepackage[T1]{fontenc}
\usepackage{lmodern}

\theoremstyle{thmstyleone}%
\newtheorem{theorem}{Theorem}
\newtheorem{proposition}[theorem]{Proposition}%

\theoremstyle{thmstyletwo}%

\theoremstyle{thmstylethree}%

\usepackage{hyperref}

\usepackage{enumitem}

\newtheorem{lemma}[theorem]{Lemma}
\newtheorem{corollary}[theorem]{Corollary}

\usepackage{thmtools,thm-restate}

\newenvironment{claimproof}[1][Proof of claim]{%
  \begin{proof}[#1]}{\end{proof}}

\DeclareMathOperator{\rad}{rad}
\DeclareMathOperator{\MP}{mp}

\begin{document}

\title[On the Complexity of Multipacking]
{On the Complexity of Multipacking}








\author[1]{\fnm{Sandip} \sur{Das}}

\author[1]{\fnm{Sk Samim} \sur{Islam}}

\author[2]{\fnm{Daniel} \sur{Lokshtanov}}

\affil[1]{\orgname{Indian Statistical Institute},
  \orgaddress{\city{Kolkata}, \country{India}}}
\affil[2]{\orgname{University of California, Santa Barbara},
  \orgaddress{\state{CA}, \country{USA}}}


\abstract{ A \emph{multipacking} in an undirected
	graph $ G = (V, E) $ is a set $ M \subseteq V  $  such that for every vertex $ v \in V $ and for every integer $ r \geq 1 $, the
	ball of radius $ r $ around $ v $ contains at most $ r $ vertices of $ M $, that is, there are at most
	$ r $ vertices in $ M $ at a distance at most $ r $ from $ v $ in $ G $. The
	\textsc{Multipacking} problem asks whether a graph contains a multipacking of size at least $k$.

    For more than a decade, it remained an open question whether the \textsc{Multipacking} problem  is \textsc{NP}-complete  or solvable in polynomial time, although the problem is known to be polynomial-time solvable for certain graph classes (e.g., strongly chordal graphs, grids, etc). Foucaud, Gras, Perez, and Sikora~\cite{foucaud2021complexity} [\textit{Algorithmica} 2021] made a step towards solving the open question by showing that the \textsc{Multipacking} problem  is \textsc{NP}-complete  for directed graphs and \textsc{W[1]-hard} when parameterized by the solution size. 

    In this paper, we prove that the \textsc{Multipacking} problem  is \textsc{NP}-complete on undirected graphs, which answers the open question. Moreover, the problem is \textsc{W[2]}-hard on undirected graphs when parameterized by the solution size. Furthermore, we show that the problem is \textsc{NP}-complete and \textsc{W[2]}-hard (parameterized by solution size) on chordal, bipartite, and claw-free graphs, and remains \textsc{NP}-complete on regular and CONV graphs (intersection graphs of convex sets in the plane). Additionally, the problem is \textsc{NP}-complete  and \textsc{W[2]}-hard  (parameterized by the solution size) on chordal $\cap$ $\frac{1}{2}$-hyperbolic graphs, which is a superclass of strongly chordal graphs on which the problem is polynomial-time solvable. On the positive side, we present an exact exponential-time algorithm for the \textsc{Multipacking} problem on general graphs that breaks the $2^n$ barrier, with running time $O(1.53^n)$, where $n$ is the number of vertices.}

\keywords{Multipacking, Computational complexity,
Parameterized complexity, Exact exponential algorithms, Graph algorithms}




\maketitle

\noindent\textbf{Relationship to the conference version.}
A preliminary version of this work appeared in the proceedings of the
34th Annual European Symposium on Algorithms (ESA 2026)~\cite{das2025hardnessMultipackingESA}.
The present version substantially extends the conference version.
In particular, the exact exponential-time algorithm has been improved from
$O^*(1.58^n)$ to $O(1.53^n)$, together with a new and more detailed analysis
of the branching procedure leading to this improved running time.

\section{Introduction}
\label{sec:introduction}

Multipacking was introduced by Brewster, Mynhardt, and Teshima~\cite{brewster2013new} as a natural generalization of classical packing and independence concepts in graphs (see also~\cite{beaudou2019multipacking,cornuejols2001combinatorial,haynes2021structures,dunbar2006broadcasts,meir1975relations,teshima2012broadcasts}). Intuitively, multipacking can be viewed as a refined version of the independent set problem, where the feasibility of selecting vertices depends on distance constraints at all scales.

For a graph $G=(V,E)$, let $d_G(u,v)$ denote the length of a shortest path between two vertices $u$ and $v$ in $G$; we write $d(u,v)$ when there is no ambiguity. Let $N_r[v]:=\{u\in V(G):d(v,u)\leq r\}$, i.e. a ball of radius $r$ around $v$.  The \textit{eccentricity} $e(w)$  of a vertex $w$ is $\min \{r:N_r[w]=V\}$. The \textit{radius} of the graph $G$ is $\min\{e(w):w\in V\}$, denoted by $\rad(G)$.  The \textit{center} $C(G)$ of the graph $G$  is the set of all vertices of minimum eccentricity, i.e., $C(G):=\{v\in V:e(v)=\rad(G)\}$. Each vertex in the set $C(G)$ is called a \textit{central vertex} of the graph $G$. Let $\mathrm{diam}(G):=\max\{d(u,v):u,v\in V(G)\}$. A diameter of $G$ is a path of length $\mathrm{diam}(G)$.

A \textit{multipacking} is a set $ M \subseteq V  $ in a
	graph $ G = (V, E) $ such that   $|N_r[v]\cap M|\leq r$ for each vertex $ v \in V $ and for every integer $ r \geq 1 $.  It is worth noting that in this definition, it suffices to consider values of $r$ in the set $\{1,2,\dots, \mathrm{diam}(G)\}$. Moreover, if $c$ is a central vertex of the graph $G$, then $N_r[c]=V(G)$ for $r=\mathrm{rad}(G)$. Therefore, the size of any multipacking in $G$ is upper-bounded by $\mathrm{rad}(G)$. Therefore, we can further restrict $r$ (in the definition of multipacking) to the set $\{1,2,\dots, \mathrm{rad}(G)\}$, where $\mathrm{rad}(G)$ is the radius of the graph. The \textit{multipacking number} of $ G $ is the maximum cardinality of a multipacking of $ G $ and it
	is denoted by $ \MP(G) $. A \textit{maximum multipacking} is a multipacking $M$  of a graph $ G  $ such that	$|M|=\MP(G)$.


\medskip
\noindent
\fbox{%
  \begin{minipage}{\dimexpr\linewidth-2\fboxsep-2\fboxrule}
  \textsc{ Multipacking} problem
  \begin{itemize}[leftmargin=1.5em]
    \item \textbf{Input:} An undirected graph $G = (V, E)$, an integer $k \in \mathbb{N}$.
    \item \textbf{Question:} Does there exist a multipacking $M \subseteq V$ of $G$ of size at least $k$?
  \end{itemize}
  \end{minipage}%
}
\medskip

It remained an open problem for over a decade whether the \textsc{Multipacking} problem  is \textsc{NP}-complete  or solvable in polynomial time. This question was repeatedly addressed by numerous authors (\cite{brewster2019broadcast,foucaud2021complexity,haynes2021structures,teshima2012broadcasts,yang2015new,yang2019limited}),  yet the problem remained as an unsolved challenge. In this paper, we answer this question by proving that the \textsc{Multipacking} problem  is \textsc{NP}-complete  which is one of the main results of this paper. 

On the algorithmic side, Beaudou, Brewster, and Foucaud~\cite{beaudou2019broadcast} presented a $(2+o(1))$-approximation algorithm for general graphs. Polynomial-time algorithms are known for trees and, more generally, for strongly chordal graphs~\cite{brewster2019broadcast}, although even the case of trees requires nontrivial techniques. Approximation algorithms with ratio $(\tfrac{3}{2}+o(1))$ have been obtained for chordal graphs~\cite{das2023relation} and cactus graphs~\cite{das2025multipacking}.

In 2021, Foucaud, Gras, Perez, and Sikora~\cite{foucaud2021complexity} studied the complexity of the \textsc{Multipacking} problem on digraphs (directed graphs). They showed that the  \textsc{Multipacking} problem is \textsc{NP}-complete  for digraphs, even for planar layered acyclic digraphs of bounded maximum degree. They further proved that the problem is \textsc{W[1]-hard} (even on digraphs of maximum finite distance $3$) when parameterized by the solution size.

It is well known that the \textsc{Multipacking} problem admits a natural integer programming formulation. If $M$ is a multipacking, we define   a vector $y$ with the entries $y_j=1$ when $v_j\in M$ and $y_j=0$ when $v_j\notin M$. Let $A=[a_{j,(i,k)}]$ be a matrix with entries $a_{j,(i,k)}=1$ if $v_j\in N_k[v_i]$, and $a_{j,(i,k)}=0$ otherwise. Let $c$ be a vector indexed by $(i,k)$ where $v_i\in V(G)$ and $1\leq k\leq \mathrm{diam}(G)$, with the  entries $c_{i,k}=k$.
 Thus,  \begin{center}
	    $\MP(G)=\max \{y\cdot\mathbf{1} :  yA\leq c, y_{j}\in \{0,1\}\}.$ 
\end{center}

A comprehensive survey of multipacking and related problems can be found in the book~\cite{haynes2021structures}, and a geometric variant of the problem is studied in~\cite{das2025geometrymultipackingCALDAM}.

 The \textit{maximum multipacking problem} is the dual integer program of the \textit{optimal broadcast domination problem}.  Erwin~\cite{erwin2004dominating,erwin2001cost} introduced broadcast domination in his doctoral thesis in
2001. For a graph $ G = (V, E) $, a function
	$ f : V \rightarrow \{0, 1, 2, . . . , \mathrm{diam}(G)\} $ is called a \textit{broadcast} on $ G $. For each
	vertex $ u \in V  $, if  there exists a vertex $ v $ in $ G $ (possibly, $ u = v $) such that $ f (v) > 0 $ and
	$ d(u, v) \leq f (v) $, then $ f $ is called a \textit{dominating broadcast} on $ G $.
The \textit{cost} of the broadcast $f$ is the quantity $ \sigma(f)  $, which is the sum
	of the weights of the broadcasts over all vertices in $ G $. So, $\sigma(f)=  \sum_{v\in V}f(v)$. The minimum cost of a dominating broadcast in $G$ (taken over all dominating broadcasts)  is the \textit{broadcast domination number} of $G$, denoted by $ \gamma_{b}(G) $. An \textit{optimal dominating broadcast} on a graph $G$ is a dominating broadcast with a cost equal to $ \gamma_{b}(G) $.


A standard integer linear programming formulation of broadcast domination is obtained as follows. Let $c$ and $x$ be the vectors indexed by $(i,k)$ where $v_i\in V(G)$ and $1\leq k\leq \mathrm{diam}(G)$, with the  entries $c_{i,k}=k$ and $x_{i,k}=1$ when $f(v_i)=k$ and $x_{i,k}=0$ when $f(v_i)\neq k$. Let $A=[a_{j,(i,k)}]$ be a matrix with entries $a_{j,(i,k)}=1$ if $v_j\in N_k[v_i]$, and $a_{j,(i,k)}=0$ otherwise.
 	Hence, the broadcast domination number can be expressed as an integer linear program:  $$\gamma_b(G)=\min \{c.x :  Ax\geq \mathbf{1}, x_{i,k}\in \{0,1\}\}.$$	

   
 For general graphs, surprisingly, an optimal dominating broadcast can be found in polynomial time $O(n^6)$~\cite{heggernes2006optimal}. The same problem can be solved in linear time for trees~\cite{brewster2019broadcast}.

 It is known that $\MP(G)\leq \gamma_b(G)$~\cite{brewster2013new}. In 2019, Beaudou,  Brewster, and Foucaud~\cite{beaudou2019broadcast} proved that $\gamma_b(G)\leq 2\MP(G)+3$ and they conjectured that $\gamma_b(G)\leq 2\MP(G)$ for every graph.  In 2025, Rajendraprasad et al.~\cite{rajendraprasad2025multipacking} have shown that for general connected graphs, 
\begin{center}
    $\displaystyle\lim_{\MP(G)\to \infty}\sup\Bigg\{\dfrac{\gamma_{b}(G)}{\MP(G)}\Bigg\}= 2.$
\end{center} 
Stronger bounds are known for specific graph classes. In particular, for any connected chordal graph $G$, we have $\gamma_b(G)\leq \left\lceil \tfrac{3}{2}\MP(G)\right\rceil$~\cite{das2023relation}, and for any cactus graph $G$, it holds that $\gamma_b(G)\leq \tfrac{3}{2}\MP(G)+\tfrac{11}{2}$~\cite{das2025multipacking}.

\medskip
\noindent\textbf{Our Contributions.}
The algorithmic complexity of the \textsc{Multipacking} problem remained open for more than a decade
(\cite{brewster2019broadcast,foucaud2021complexity,teshima2012broadcasts,yang2015new,yang2019limited}). In this paper, we resolve this question by establishing hardness results for \textsc{Multipacking} on a wide range of graph classes, complemented by an exact exponential-time algorithm. 

Our main focus is to identify graph classes for which the \textsc{Multipacking} problem admits a polynomial-time algorithm and to delineate a boundary between tractable and intractable cases. 
Whenever our reductions also imply consequences for the parameterized complexity of \textsc{Multipacking} on the considered graph classes, we explicitly point them out. 
However, we do not aim at obtaining tight parameterized complexity bounds, and we deliberately avoid more involved reductions whose sole purpose would be to strengthen such bounds.

Our main hardness result is obtained via a simple and elegant reduction from the \textsc{Hitting Set} problem, which is \textsc{NP}-complete ~\cite{garey1979computers} and \textsc{W[2]}-complete  when parameterized by the solution size~\cite{cygan2015parameterized}. This reduction yields the following theorem.

\begin{restatable}{theorem}{MultipackingHardnessgeneral}\label{thm:Multipacking_general_hardness}
The \textsc{Multipacking} problem is \textsc{NP}-complete. Moreover, when parameterized by the solution size $k$, the problem is \textsc{W[2]}-hard. In addition, unless the \textit{Exponential Time Hypothesis (ETH)}\footnote{The \textit{Exponential Time Hypothesis (ETH)} states that \textsc{3-SAT} cannot be solved in time $2^{o(n)}$, where $n$ is the number of variables of the input CNF formula.} fails, there is no algorithm for \textsc{Multipacking} running in time $f(k)n^{o(k)}$, where $n$ is the number of vertices of the input graph.
\end{restatable}

We further study the complexity of \textsc{Multipacking} on graph classes defined by metric properties, with a particular emphasis on $\frac{1}{2}$-hyperbolic graphs. The notion of $\delta$-hyperbolicity was introduced by Gromov~\cite{gromov1987hyperbolic}. A formal definition of hyperbolic graphs is provided in the Preliminaries (Section~\ref{sec:preliminaries}). While \textsc{Multipacking} is solvable in polynomial time on strongly chordal graphs\footnote{A \emph{strongly chordal graph} is  an undirected graph $G$ which is a chordal graph and every cycle of even length ($\geq 6$) in $G$ has an \emph{odd chord}, i.e., an edge that connects two vertices that are an odd distance ($>1$) apart from each other in the cycle.} and in linear time on trees~\cite{brewster2019broadcast}, and while strongly chordal graphs are $\frac{1}{2}$-hyperbolic~\cite{wu2011hyperbolicity}, we show that this tractability does not extend to a larger class chordal $\cap$ $\frac{1}{2}$-hyperbolic. In particular, we strengthen the above result by proving that \textsc{Multipacking} remains \textsc{NP}-complete  even for chordal $\cap$ $\frac{1}{2}$-hyperbolic graphs, a strict superclass of strongly chordal graphs.

Beyond metric graph classes, we establish hardness results for several classical families. Using two distinct reductions from \textsc{Hitting Set}, we show \textsc{NP}-completeness and \textsc{W[2]}-hardness for bipartite graphs and claw-free\footnote{\textit{Claw-free graph} is a graph that does not contain a claw ($K_{1,3}$) as an induced subgraph.} graphs. We also prove \textsc{NP}-completeness for regular graphs via a reduction from the \textsc{Total Dominating Set} problem on cubic (3-regular) graphs~\cite{garey1979computers}. Furthermore, we study \textsc{Multipacking} on geometric intersection graphs: we show that the problem is \textsc{NP}-complete  for CONV graphs\footnote{A graph $G$ is a \textit{CONV graph} or \textit{Convex Intersection Graph} if and only if there exists a family of convex sets on a plane such that the graph has a vertex for each convex set and an edge for each intersecting pair of convex sets.}  by a reduction from \textsc{Total Dominating Set} on planar graphs.

These results are summarized in the following theorem.

\begin{restatable}{theorem}{MultipackingHardnesssubclass}\label{thm:Multipacking_hardness_subclass}
The \textsc{Multipacking} problem is \textsc{NP}-complete  for each of the following graph classes (See Fig.~\ref{fig:graph_class_map}):
\begin{itemize}
    \item chordal $\cap$ $\frac{1}{2}$-hyperbolic graphs,
    \item bipartite graphs,
    \item claw-free graphs,
    \item regular graphs, and
    \item CONV (convex intersection) graphs.
\end{itemize}
Moreover, when parameterized by the solution size, the problem is \textsc{W[2]}-hard  for chordal $\cap$ $\frac{1}{2}$-hyperbolic graphs, bipartite graphs, and claw-free graphs.
\end{restatable}

\begin{figure}[ht]
    \centering
   \includegraphics[width=\textwidth]{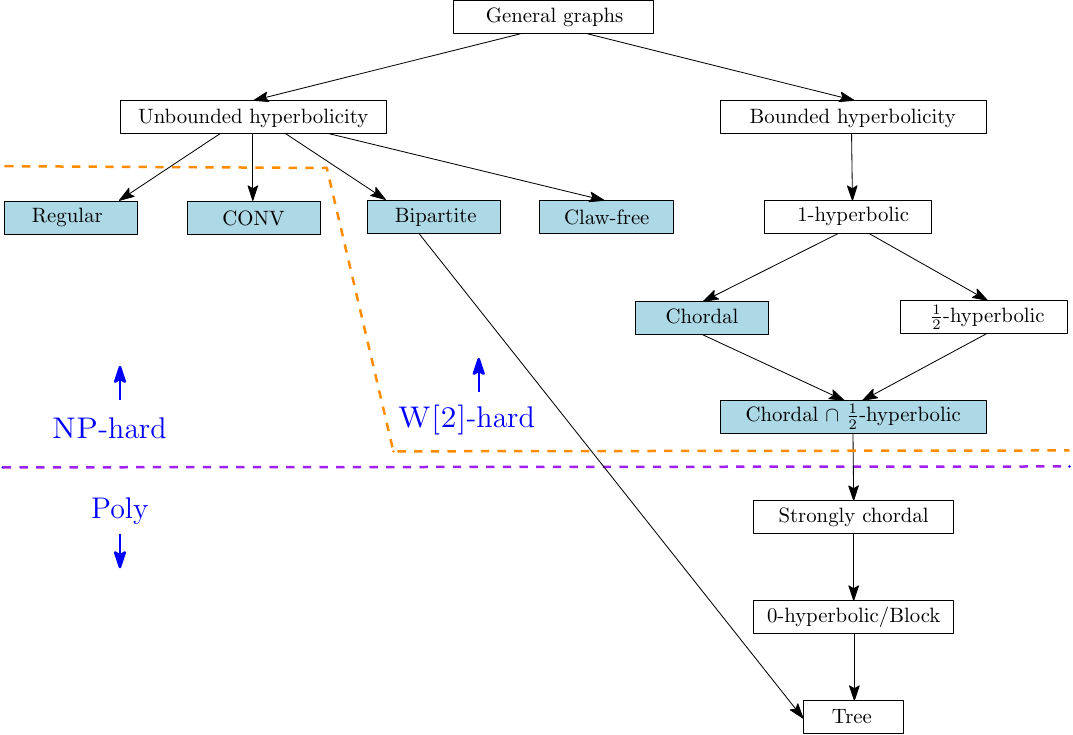}
    \caption{Inclusion diagram for graph classes mentioned in this paper (and related ones). If a class $A$ has a downward path to class $B$, then $B$ is a subclass of $A$. The \textsc{Multipacking} problem is \textsc{NP}-hard for the graph classes above the dashed purple line and polynomial-time solvable for the graph classes below it. Moreover, the problem  is \textsc{W[2]}-hard  for the graph classes above the dashed darkorange curve. In this paper, we discuss the hardness of the \textsc{Multipacking} problem for the graph classes in the colored (lightblue) block.}
    \label{fig:graph_class_map}
\end{figure}

Finally, complementing our hardness results, we design an exact exponential-time algorithm for \textsc{Multipacking} on general graphs.

\begin{restatable}{theorem}{multipackingtime}\label{thm:1.53time}
There exists an exact algorithm for computing a maximum multipacking in an $n$-vertex graph with running time $O(1.53^n)$.
\end{restatable}



\medskip
\noindent\textbf{Organisation:} In Section \ref{sec:preliminaries}, we recall some definitions and notations.  In Section \ref{sec:multipacking_npc_w}, we prove that the \textsc{Multipacking} problem is \textsc{NP}-complete  and also the problem is \textsc{W[2]}-hard  when parameterized by the solution size. In Section \ref{sec:hardness_on_subclasses}, we discuss the hardness of the \textsc{Multipacking} problem on some subclasses. In Section \ref{sec:exponentialalgorithm}, we present an exact exponential algorithm for the problem on general graphs. We conclude in Section \ref{sec:conclusion} by presenting several important questions and future research directions.

\section{Preliminaries}\label{sec:preliminaries}

A complete bipartite graph $K_{1,3}$ is called a \textit{claw}, and the vertex that is incident with all the edges in a claw is called the \textit{claw-center}. A \textit{Claw-free graph} is a graph that does not contain a claw as an induced subgraph. 



Let $G$ be a connected graph. For any four
vertices $u, v, x, y \in V(G)$, we define $\delta(u, v, x, y)$ as half of the difference between the two larger of the three sums $d(u, v) + d(x,y)$, $d(u, y) + d(v, x)$, $d(u, x) +
d(v, y)$. The \textit{hyperbolicity} of a graph $G$, denoted by $\delta(G)$, is the value of $\max_{u, v, x, y\in V(G)}\delta(u, v, x, y)$. For every $\delta\geq \delta(G)$, we say that $G$ is $\delta$-\textit{hyperbolic}. A graph class $\mathcal{G}$ is said to be \textit{hyperbolic} (or bounded hyperbolicity class) if there exists a constant $\delta$ such that every graph $G \in \mathcal{G}$ is $\delta$-hyperbolic.

The complement of a graph $G$, denoted by $G^c$, is the graph with vertex set $V(G^c)=V(G)$, where two distinct vertices are adjacent in $G^c$ if and only if they are non-adjacent in $G$. Two graphs $G$ and $H$ are said to be \emph{isomorphic} if there exists a bijection $f:V(G)\to V(H)$ such that $uv\in E(G)$ if and only if $f(u)f(v)\in E(H)$ for all $u,v\in V(G)$.


\section{Hardness of the \textsc{Multipacking} problem (Proof of Theorem~\ref{thm:Multipacking_general_hardness})}\label{sec:multipacking_npc_w}

In this section, we present a reduction from the \textsc{Hitting Set} problem to prove our first hardness result for the \textsc{Multipacking} problem.

\begin{restatable}{lemma}{NPCchordal}\label{thm:NPC_chordal} The 
 \textsc{Multipacking} problem is \textsc{NP}-complete. Moreover, the problem  is \textsc{W[2]}-hard  when parameterized by the solution size.
\end{restatable}

\begin{proof}

We reduce the well-known \textsc{NP}-complete  (also \textsc{W[2]}-complete  when parametrized by solution size) problem \textsc{Hitting Set}~\cite{cygan2015parameterized,garey1979computers} to the \textsc{Multipacking} problem.

\medskip
\noindent
\fbox{%
  \begin{minipage}{\dimexpr\linewidth-2\fboxsep-2\fboxrule}
  \textsc{ Hitting Set} problem
  \begin{itemize}[leftmargin=1.5em]
    \item \textbf{Input:} A finite set $U$, a collection $\mathcal{F}$ of subsets of $U$, and an integer $k \in \mathbb{N}$.
    \item \textbf{Question:} Does there exist a \emph{hitting set} $H \subseteq U$ of size at most $k$; that is, a set of at most $k$ elements from $U$ such that each set in $\mathcal{F}$ contains at least one element from $H$?
  \end{itemize}
  \end{minipage}%
}
\medskip

Let $U=\{u_1,u_2,\dots,u_n\}$ and  $\mathcal{F}=\{S_1,S_2,\dots,S_m\}$ where $S_i\subseteq U$ for each $i$. We construct a graph $G$ in the following way: (i) Include each element of $\mathcal{F}$ in the vertex set of $G$ and $\mathcal{F}$ forms a clique in $G$. (ii) Include $k-2$ length path $u_i = u_i^1 u_i^2 \dots u_i^{k-1}$  in $G$,  for each element $u_i\in U$.
(iii) Join $u_i^1$ and $S_j$ by an edge iff $u_i \notin S_j$. Fig. \ref{fig:reduction_chordal} gives an illustration.

\begin{figure}[ht]
    \centering
   \includegraphics[width=0.85\textwidth]{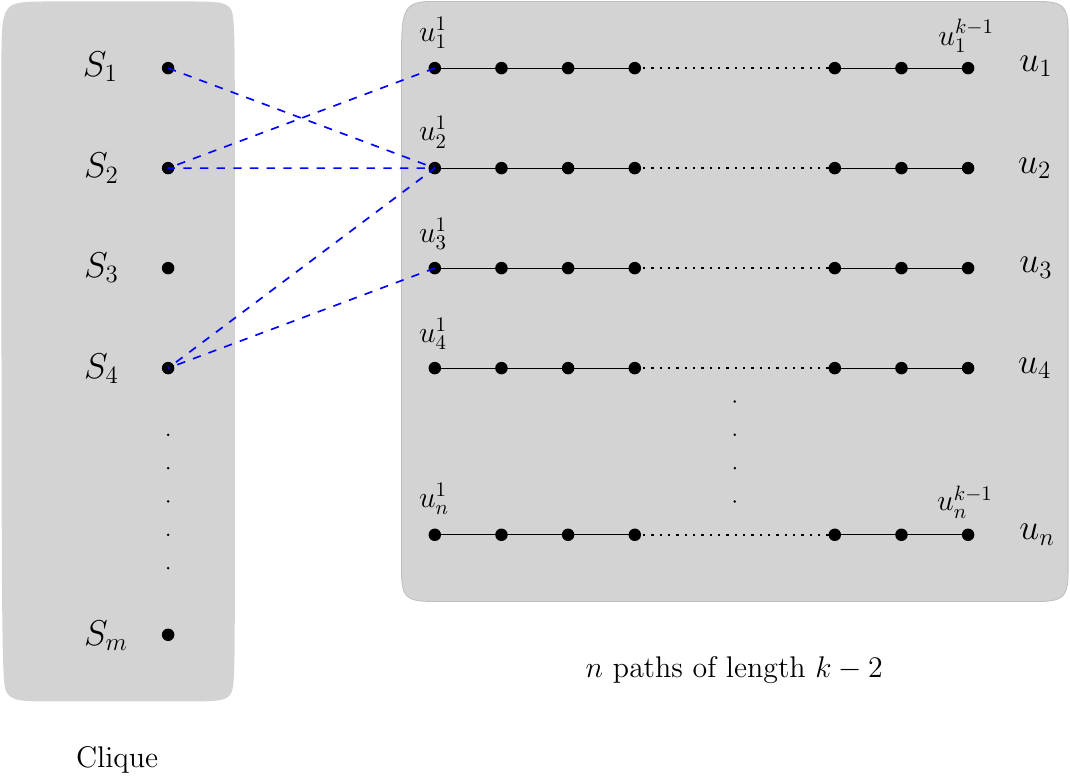}
    \caption{An illustration of the construction of the graph $G$ used in the proof of Lemma~\ref{thm:NPC_chordal}, demonstrating the hardness of the \textsc{Multipacking} problem.}
    \label{fig:reduction_chordal}
\end{figure}


\vspace{0.3cm}
\noindent
\textbf{Claim \ref{thm:NPC_chordal}.1. } $\mathcal{F}$ has a hitting set of size at most $k$ iff $G$ has a multipacking of size at least $k$, for $k\geq 2$.

\begin{claimproof} \textsf{(If part)} Suppose $\mathcal{F}$ has a hitting set $H = \{ u_1, u_2, \dots, u_k \}$. 
We want to show that $M =\{ u_1^{k-1}, u_2^{k-1}, \dots, u_k^{k-1} \}$ is a multipacking in $G$. If $r < k-1$, then 
$| N_{r}[v] \cap M |\leq 1 \leq r$, $ \forall v \in V(G)$, since $d(u_i^{k-1}, u_j^{k-1})\geq 2k-2$, $\forall i\neq j$.
If $r > k-1$, then 
$| N_{r}[v] \cap M | \leq | M | = k \leq r$, $\forall v \in V(G)$. Suppose $r = k-1$. If $M$ is not a multipacking,  there exists some $v \in V(G)$ such that 
$| N_{k-1}[v] \cap M | = k.$
Therefore, 
$M \subseteq N_{k-1}[v].$ Suppose $v=u_p^q$ for some $p$ and $q$. In this case, $d(u_p^q,u_i^{k-1})\geq k$, $\forall i\neq p$. Therefore, $v\notin \{u_i^j:1\leq i\leq n, 1\leq j\leq k-1\}$. 
This implies that $v \in \{ S_1, S_2, \dots, S_m \}$. Let $v = S_t$ for some $t$. If there exists $i\in\{1,2,\dots,k\}$ such that $S_t$ is not adjacent to $u_i^1$, then $d(S_t,u_i^{k-1})\geq k$. Therefore, 
$M \subseteq N_{k-1}[S_t] $ implies that  $S_t$ is adjacent to each vertex of the set  $\{u_1^1, u_2^1, \dots, u_k^1\}$.
Then 
$u_1, u_2, \dots, u_k \notin S_t.$
Therefore, $H$ is not a hitting set of $\mathcal{F}$. This is a contradiction. Hence, $M$ is a multipacking of size $k$.

 \textsf{(Only-if part)}  Suppose $G$ has a multipacking $M$ of size $k$. Let 
$H = \{ u_i : u_i^j \in M, 1\leq i\leq n, 1\leq j\leq k-1 \}.$
Then 
$|H| \leq |M| = k.$
Without loss of generality, let 
$H = \{ u_1, u_2, \dots, u_{k'} \} $ where $ k' \leq k$.
Now we show that $H$ is a hitting set of $\mathcal{F}$. Suppose $H$ is not a hitting set. Then there exists $t$ such that 
$S_t \cap H = \emptyset.$
Therefore, $S_t$ is adjacent to each vertex of the set $\{ u_1^{1}, u_2^{1}, \dots, u_{k'}^{1} \}$.
Then $M \subseteq N_{k-1}[S_t]$. This implies that 
$| N_{k-1}[S_t] \cap M | = k.$
This is a contradiction. Therefore, $H$ is a hitting set of size at most $k$. 
\end{claimproof}

\noindent Since the above reduction is a polynomial-time, the \textsc{Multipacking} problem  is \textsc{NP}-complete.  Moreover, the reduction is also an FPT reduction. Hence, the \textsc{Multipacking} problem  is \textsc{W[2]}-hard  when parameterized by the solution size.
\end{proof}

\begin{corollary}\label{cor:NPC_chordal}
  \textsc{Multipacking} problem is \textsc{NP}-complete  for chordal graphs. Moreover, \textsc{Multipacking} problem  is \textsc{W[2]}-hard  for chordal graphs when parameterized by the solution size.
    
\end{corollary}

\begin{proof}
    We prove that the graph $G$ constructed in the proof of Lemma~\ref{thm:NPC_chordal} is chordal. Note that if $G$ has a cycle $C=c_0c_1c_2\dots c_{t-1}c_0$ of length at least $4$, then there exists an index $i$ such that $c_i,c_{i+2 \! \pmod{t}}\in \mathcal{F}$ because no two vertices of the set $\{ u_1^{1}, u_2^{1}, \dots, u_n^{1} \}$ are adjacent. Since $\mathcal{F}$ forms a clique, so $c_i$ and $c_{i+2 \! \pmod{t}}$ are endpoints of a chord. Therefore, the graph $G$ is chordal.
\end{proof}


It is known that, unless \textsc{ETH} fails, there is no $f(k)(m+n)^{o(k)}$-time algorithm for the \textsc{Hitting Set} problem, where $k$ is the solution size, $|U|=n$ and  $|\mathcal{F}|=m$ \cite{cygan2015parameterized}.

\begin{restatable}{lemma}{ETHmultipackingSubExp}\label{thm:ETH_multipacking_fkno(k)}
    Unless \textsc{ETH} fails, there is no $f(k)n^{o(k)}$-time algorithm for the \textsc{Multipacking} problem, where $k$ is the solution size and $n$ is the number of vertices of the graph.
\end{restatable}

\begin{proof} Consider the construction of $G$ in the proof of Lemma \ref{thm:NPC_chordal}.
   Let $|V(G)|=n_1$. From the construction, we have $n_1=m+n(k-1)=O(m+n^2)$. Therefore, the construction of $G$ takes $O(m+n^2)$-time. Suppose that there exists an $f(k)n_1^{o(k)}$-time algorithm for the \textsc{Multipacking} problem. Therefore, we can solve the \textsc{Hitting Set} problem in $f(k)n_1^{o(k)}+O(m+n^2)=f(k)(m+n^2)^{o(k)}=f(k)(m+n)^{o(k)}$ time. This is a contradiction, since there is no $f(k)(m+n)^{o(k)}$-time algorithm for the \textsc{Hitting Set} problem, assuming \textsc{ETH}.  
\end{proof}

\section{Hardness results on several subclasses (Proof of Theorem~\ref{thm:Multipacking_hardness_subclass})} \label{sec:hardness_on_subclasses}

In this section, we present hardness results on some subclasses.

\subsection{Chordal $\cap \text{ } \frac{1}{2}$-hyperbolic graphs}

In Corollary \ref{cor:NPC_chordal},  we have already shown that, for chordal graphs, the \textsc{Multipacking} problem is \textsc{NP}-complete  and the problem is \textsc{W[2]}-hard  when parameterized by the solution size. Here, we strengthen this result. It is known that the \textsc{Multipacking} problem is solvable in cubic time for strongly chordal graphs~\cite{brewster2019broadcast}. We prove hardness for chordal $\cap \text{ } \frac{1}{2}$-hyperbolic graphs, a superclass of strongly chordal graphs~\cite{wu2011hyperbolicity}. To show this, we need the following result.

\begin{theorem}[\cite{brinkmann2001hyperbolicity}]\label{thm:chordal_one_hyperbolic}
 If $G$ is a chordal graph, then its hyperbolicity  $\delta(G)\leq 1$. Moreover, $\delta(G)=1$  if and only if  $G$ contains one of the graphs in Fig. \ref{fig:chordal_half_hyperbolic_forbidden_graphs} as an isometric subgraph\footnote{A subgraph $H$ of a graph $G$ is \emph{isometric} if for any $u,v \in V(H)$, it holds that $d_H(u,v)=d_G(u,v)$.}.  Equivalently, a chordal graph $G$ is $\frac{1}{2}$-hyperbolic (i.e., $\delta(G)\leq \frac{1}{2}$) if and only if $G$ does not contain any of the graphs in Fig.~\ref{fig:chordal_half_hyperbolic_forbidden_graphs} as an isometric subgraph.
\end{theorem}

\begin{figure}[ht]
    \centering
    \includegraphics[height=3.5cm]{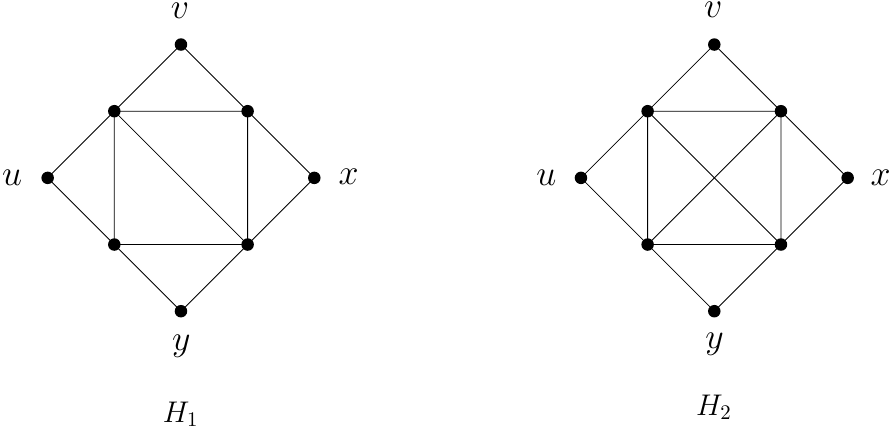}
    \caption{Forbidden isometric subgraphs for chordal $\cap$ $\frac{1}{2}$-hyperbolic graphs} 
    \label{fig:chordal_half_hyperbolic_forbidden_graphs}
\end{figure}



\begin{restatable}{lemma}{NPChalfhyperbolicchordal}\label{thm:NPC_half_hyperbolic_chordal}
  \textsc{Multipacking} problem  is \textsc{NP}-complete  for chordal $\cap$ $\frac{1}{2}$-hyperbolic graphs. Moreover, \textsc{Multipacking} problem  is \textsc{W[2]}-hard  for chordal $\cap$ $\frac{1}{2}$-hyperbolic graphs when parameterized by the solution size.
\end{restatable}



\begin{proof} We reduce the \textsc{Hitting Set} problem to the \textsc{Multipacking} problem  to show that the \textsc{Multipacking} problem  is \textsc{NP}-complete  for chordal $\cap$ $\frac{1}{2}$-hyperbolic graphs. 

Let $U=\{u_1,u_2,\dots,u_n\}$ and  $\mathcal{F}=\{S_1,S_2,\dots,S_m\}$ where $S_i\subseteq U$ for each $i$. We construct a graph $G$ in the following way: (i) Include each element of $\mathcal{F}$ in the vertex set of $G$. (ii) Include $k-2$ length path $u_i = u_i^1 u_i^2 \dots u_i^{k-1}$  in $G$,  for each element $u_i\in U$.
(iii) Join $u_i^1$ and $S_j$ by an edge iff $u_i \notin S_j$. (iv) Include the vertex set  $Y=\{y_{i,j}:1\leq i<j\leq n\}$ in $G$. (v) For each $i,j$ where $1\leq i<j\leq n$, $y_{i,j}$ is adjacent to $u_i^1$ and $u_j^1$. (vi) $\mathcal{F}\cup Y$ forms a clique in $G$. Fig. \ref{fig:reduction_chordal_half_hyperbolic} gives an illustration.  

\begin{figure}[ht]
    \centering
   \includegraphics[width=0.85\textwidth]{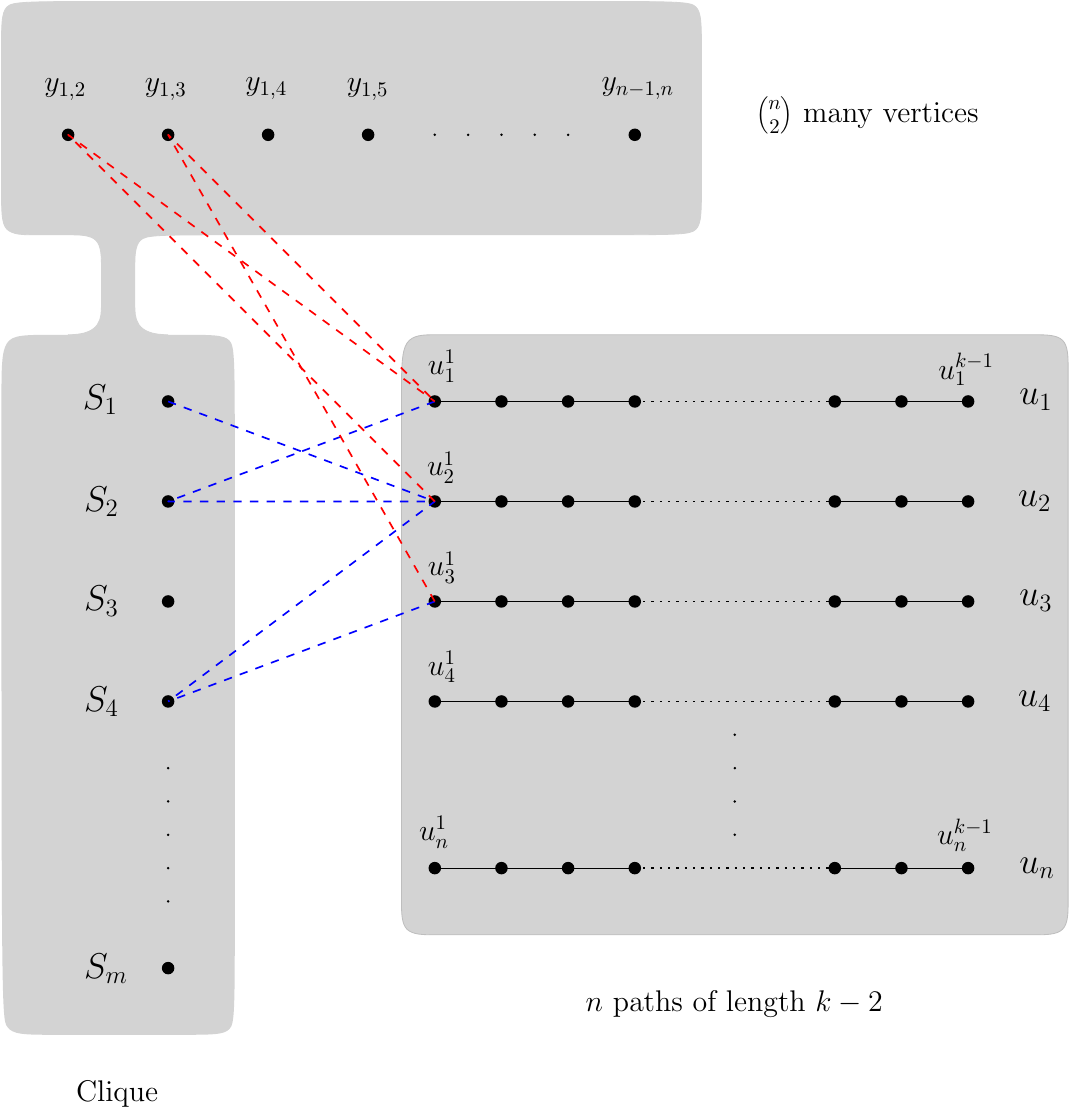}
    \caption{An illustration of the construction of the graph $G$ used in the proof of Lemma~\ref{thm:NPC_half_hyperbolic_chordal}, demonstrating the hardness of the \textsc{Multipacking} problem for chordal $\cap$ $\frac{1}{2}$-hyperbolic graphs.}
    \label{fig:reduction_chordal_half_hyperbolic}
\end{figure}

\vspace{0.3cm}
\noindent
\textbf{Claim \ref{thm:NPC_half_hyperbolic_chordal}.1. } $G$ is a chordal graph.

\begin{claimproof} If $G$ has a cycle $C=c_0c_1c_2\dots c_{t-1}c_0$ of length at least $4$, then there exists an index $i$ such that $c_i,c_{i+2 \! \pmod{t}}\in \mathcal{F}\cup Y$ because no two vertices of the set $\{ u_1^{1}, u_2^{1}, \dots, u_n^{1} \}$ are adjacent. Since $\mathcal{F}\cup Y$ forms a clique, $c_i$ and $c_{i+2 \! \pmod{t}}$ are endpoints of a chord. Therefore, the graph $G$ is chordal.    
\end{claimproof}

\vspace{0.3cm}
\noindent
\textbf{Claim \ref{thm:NPC_half_hyperbolic_chordal}.2. } $G$ is a  $\frac{1}{2}$-hyperbolic graph.

\begin{claimproof} Let $G'$ be the induced subgraph of $G$ on the vertex set $\mathcal{F}\cup Y\cup \{ u_1^{1}, u_2^{1}, \dots, u_n^{1} \}$. We want to show that $\mathrm{diam}(G')\leq 2$. Note that the distance between any two vertices of the set $\mathcal{F}\cup Y$ is $1$, since $\mathcal{F}\cup Y$ forms a clique in $G$. Furthermore, the distance between any two vertices of the set $\{ u_1^{1}, u_2^{1}, \dots, u_n^{1} \}$ is $2$, since they are non-adjacent and have a common neighbor in the set $Y$. Suppose $v\in \mathcal{F}\cup Y$ and $u_i^1\in \{ u_1^{1}, u_2^{1}, \dots, u_n^{1} \}$ where $v$ and $u_i^1$ are non-adjacent. In that case, $u_iy_{1,i}v$ is a $2$-length path that joins $v$ and $u_i^1$ since $\mathcal{F}\cup Y$ forms a clique in $G$. Hence, $\mathrm{diam}(G')\leq 2$. By Claim \ref{thm:NPC_half_hyperbolic_chordal}.1 we have that $G$ is chordal. From Theorem \ref{thm:chordal_one_hyperbolic}, we can say that the hyperbolicity of $G$ is at most $1$. Suppose $G$ has hyperbolicity exactly $1$. In that case, from Theorem \ref{thm:chordal_one_hyperbolic}, we can say that $G$  contains at least one of the graphs $H_1$ or $H_2$  (Fig. \ref{fig:chordal_half_hyperbolic_forbidden_graphs}) as 
 isometric subgraphs. Note that every vertex of $H_1$ and $H_2$ is a vertex of a cycle. Therefore, $H_1$ or $H_2$ does not have a vertex from the set $\{u_i^j:1\leq i\leq n, 2\leq j\leq k-1\}$. This implies that $H_1$ or $H_2$ are isometric subgraphs of $G'$. But both $H_1$ and $H_2$ have diameter $3$, while $\mathrm{diam}(G')\leq 2$. This is a contradiction. Hence, $G$ is a  $\frac{1}{2}$-hyperbolic graph.   
\end{claimproof}

\vspace{0.3cm}
\noindent
\textbf{Claim \ref{thm:NPC_half_hyperbolic_chordal}.3. } $\mathcal{F}$ has a hitting set of size at most $k$ iff $G$ has a multipacking of size at least $k$, for $k\geq 3$.

\begin{claimproof} \textsf{(If part)} Suppose $\mathcal{F}$ has a hitting set $H = \{ u_1, u_2, \dots, u_k \}$. 
We want to show that $M =\{ u_1^{k-1}, u_2^{k-1}, \dots, u_k^{k-1} \}$ is a multipacking in $G$. If $r < k-1$, then 
$| N_{r}[v] \cap M |\leq 1 \leq r$, $ \forall v \in V(G)$, since $d(u_i^{k-1}, u_j^{k-1})\geq 2k-2$, $\forall i\neq j$.
If $r > k-1$, then 
$| N_{r}[v] \cap M | \leq | M | = k \leq r$, $\forall v \in V(G)$. Suppose $r = k-1$. If $M$ is not a multipacking,  there exists some $v \in V(G)$ such that 
$| N_{k-1}[v] \cap M | = k.$
Therefore, 
$M \subseteq N_{k-1}[v].$ Suppose $v=u_p^q$ for some $p$ and $q$. In this case, $d(u_p^q,u_i^{k-1})\geq k$, $\forall i\neq p$. Therefore, $v\notin \{u_i^j:1\leq i\leq n, 1\leq j\leq k-1\}$. This implies that $v \in \mathcal{F}\cup Y$. Suppose $v \in Y$. Let $v=y_{p,q}$ for some $p$ and $q$. Then $d(y_{p,q},u_i^{k-1})\geq k$, $\forall i\in \{1,2,\dots,n\}\setminus\{p,q\} $. Therefore, $v\notin Y$. This implies that $v \in \mathcal{F}$. Let $v = S_t$ for some $t$. If there exist $i\in\{1,2,\dots,k\}$ such that $S_t$ is not adjacent to $u_i^1$, then $d(S_t,u_i^{k-1})\geq k$. Therefore, 
$M \subseteq N_{k-1}[S_t] $ implies that  $S_t$ is adjacent to each vertex of the set  $\{u_1^1, u_2^1, \dots, u_k^1\}$.
Then 
$u_1, u_2, \dots, u_k \notin S_t.$
Therefore, $H$ is not a hitting set of $\mathcal{F}$. This is a contradiction. Hence, $M$ is a multipacking of size $k$.

 \textsf{(Only-if part)}  Suppose $G$ has a multipacking $M$ of size $k$. Let 
$H = \{ u_i : u_i^j \in M, 1\leq i\leq n, 1\leq j\leq k-1 \}.$
Then 
$|H| \leq |M| = k.$
Without loss of generality, let 
$H = \{ u_1, u_2, \dots, u_{k'} \} $ where $ k' \leq k$.
Now we show that $H$ is a hitting set of $\mathcal{F}$. Suppose $H$ is not a hitting set. Then there exists $t$ such that 
$S_t \cap H = \emptyset.$
Therefore, $S_t$ is adjacent to each vertex of the set $\{ u_1^{1}, u_2^{1}, \dots, u_{k'}^{1} \}$.
Then $M \subseteq N_{k-1}[S_t]$. This implies that 
$| N_{k-1}[S_t] \cap M | = k.$
This is a contradiction. Therefore, $H$ is a hitting set of size at most $k$.
\end{claimproof}

\noindent Since the above reduction is a polynomial-time, the \textsc{Multipacking} problem  is \textsc{NP}-complete  for chordal $\cap$ $\frac{1}{2}$-hyperbolic graphs.  Moreover, the reduction is also an FPT reduction. Hence, the \textsc{Multipacking} problem  is \textsc{W[2]}-hard  for chordal $\cap$ $\frac{1}{2}$-hyperbolic graphs when parameterized by the solution size.
\end{proof}

\subsection{Bipartite graphs}

Here, we discuss the hardness of the \textsc{Multipacking} problem for bipartite graphs.

\begin{lemma}\label{thm:NPC_bipartite}
  \textsc{Multipacking} problem  is \textsc{NP}-complete  for bipartite graphs. Moreover, \textsc{Multipacking} problem  is \textsc{W[2]}-hard  for bipartite graphs when parameterized by the solution size.
\end{lemma}



\begin{proof}

We reduce the \textsc{Hitting Set} problem to the \textsc{Multipacking} problem  to show that the \textsc{Multipacking} problem  is \textsc{NP}-complete  for bipartite graphs. 

Let $U=\{u_1,u_2,\dots,u_n\}$ and  $\mathcal{F}=\{S_1,S_2,\dots,S_m\}$ where $S_i\subseteq U$ for each $i$. We construct a graph $G$ in the following way: (i)  Include each element of $\mathcal{F}$ in the vertex set of $G$. (ii) Include a vertex $C$ in $G$ where all the vertices of $\mathcal{F}$ are adjacent to the vertex $C$. (iii) Include $k-2$ length path $u_i = u_i^1 u_i^2 \dots u_i^{k-1}$  in $G$,  for each element $u_i\in U$.
(iv) Join $u_i^1$ and $S_j$ by an edge iff $u_i \notin S_j$.  Fig. \ref{fig:reduction_bipartite} gives an illustration. 

Note that $G$ is a bipartite graph where the partite sets are $B_1=\mathcal{F}\cup \{u_i^j:1\leq i\leq n \text{ and } j\text{ is even} \}$ and $B_2=\{C\}\cup \{u_i^j:1\leq i\leq n \text{ and } j\text{ is odd} \}$. 

\begin{figure}[ht]
    \centering
   \includegraphics[width=0.9\textwidth]{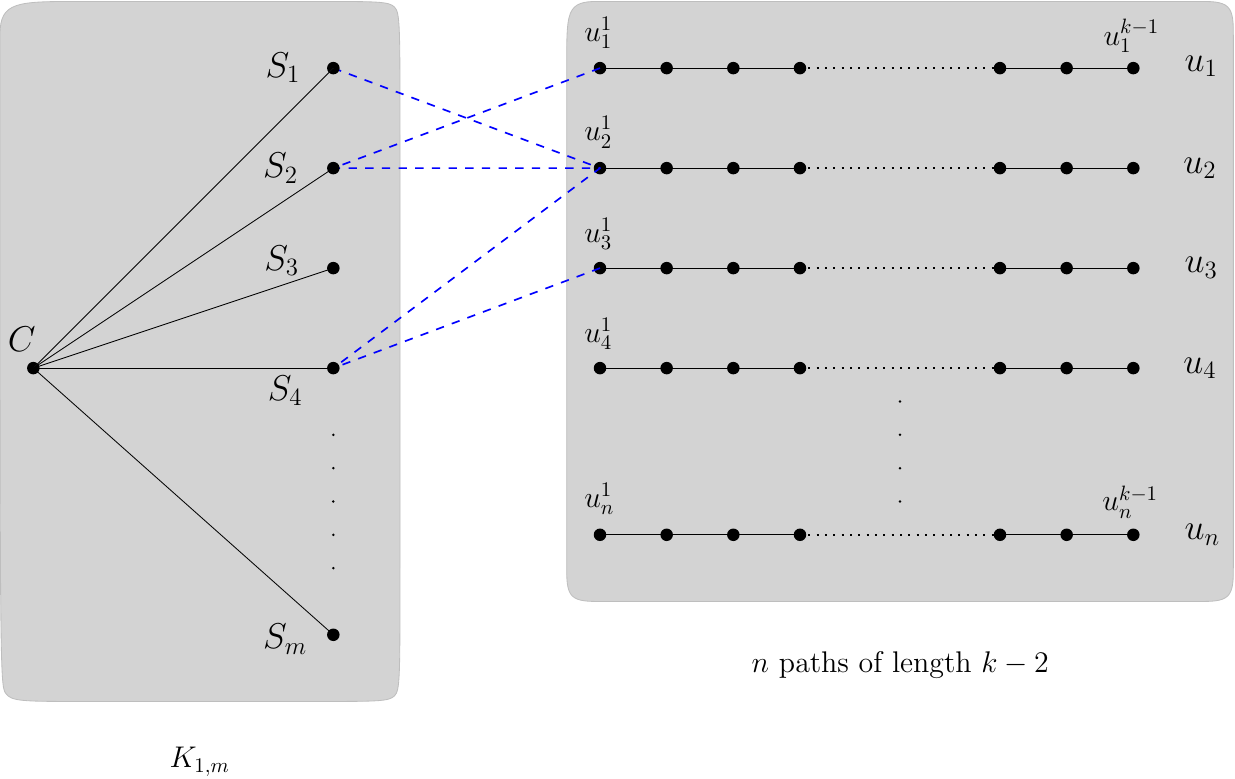}
    \caption{An illustration of the construction of the graph $G$ used in the proof of Lemma~\ref{thm:NPC_bipartite}, demonstrating the hardness of the \textsc{Multipacking} problem for bipartite graphs.}
    \label{fig:reduction_bipartite}
\end{figure}

\vspace{0.3cm}
\noindent
\textbf{Claim \ref{thm:NPC_bipartite}.1. } $\mathcal{F}$ has a hitting set of size at most $k$ iff $G$ has a multipacking of size at least $k$, for $k\geq 2$.

\begin{claimproof} \textsf{(If part)} Suppose $\mathcal{F}$ has a hitting set $H = \{ u_1, u_2, \dots, u_k \}$. 
We want to show that $M =\{ u_1^{k-1}, u_2^{k-1}, \dots, u_k^{k-1} \}$ is a multipacking in $G$. If $r < k-1$, then 
$| N_{r}[v] \cap M |\leq 1 \leq r$, $ \forall v \in V(G)$, since $d(u_i^{k-1}, u_j^{k-1})\geq 2k-2$, $\forall i\neq j$.
If $r > k-1$, then 
$| N_{r}[v] \cap M | \leq | M | = k \leq r$, $\forall v \in V(G)$. Suppose $r = k-1$. If $M$ is not a multipacking,  there exists some $v \in V(G)$ such that 
$| N_{k-1}[v] \cap M | = k.$
Therefore, 
$M \subseteq N_{k-1}[v].$ Suppose $v=u_p^q$ for some $p$ and $q$. In this case, $d(u_p^q,u_i^{k-1})\geq k$, $\forall i\neq p$. Therefore, $v\notin \{u_i^j:1\leq i\leq n, 1\leq j\leq k-1\}$. Moreover, $v\neq C$, since $d(C,u_i^{k-1})\geq k$, $\forall 1\leq i\leq n$. 
This implies that $v \in \{ S_1, S_2, \dots, S_m \}$. Let $v = S_t$ for some $t$.  If there exist $i\in\{1,2,\dots,k\}$ such that $S_t$ is not adjacent to $u_i^1$, then $d(S_t,u_i^{k-1})\geq k$. Therefore,  
$M \subseteq N_{k-1}[S_t] $ implies that  $S_t$ is adjacent to each vertex of the set  $\{u_1^1, u_2^1, \dots, u_k^1\}$.
Then 
$u_1, u_2, \dots, u_k \notin S_t.$
Therefore, $H$ is not a hitting set of $\mathcal{F}$. This is a contradiction. Hence, $M$ is a multipacking of size $k$.

 \textsf{(Only-if part)}  Suppose $G$ has a multipacking $M$ of size $k$. Let 
$H = \{ u_i : u_i^j \in M, 1\leq i\leq n, 1\leq j\leq k-1 \}.$
Then 
$|H| \leq |M| = k.$
Without loss of generality, let 
$H = \{ u_1, u_2, \dots, u_{k'} \} $ where $ k' \leq k$.
Now we show that $H$ is a hitting set of $\mathcal{F}$. Suppose $H$ is not a hitting set. Then there exists $t$ such that 
$S_t \cap H = \emptyset.$
Therefore, $S_t$ is adjacent to each vertex of the set $\{ u_1^{1}, u_2^{1}, \dots, u_{k'}^{1} \}$.
Then $M \subseteq N_{k-1}[S_t]$. This implies that 
$| N_{k-1}[S_t] \cap M | = k.$
This is a contradiction. Therefore, $H$ is a hitting set of size at most $k$. 
\end{claimproof}

\noindent Since the above reduction is a polynomial-time, the \textsc{Multipacking} problem  is \textsc{NP}-complete  for bipartite graphs.  Moreover, the reduction is also an FPT reduction. Hence, the \textsc{Multipacking} problem  is \textsc{W[2]}-hard  for bipartite graphs when parameterized by the solution size.
\end{proof}

\subsection{Claw-free graphs}

Next, we discuss the hardness of the \textsc{Multipacking} problem for claw-free graphs.

\begin{lemma}\label{thm:NPC_clawfree}
  \textsc{Multipacking} problem  is \textsc{NP}-complete  for claw-free graphs. Moreover, \textsc{Multipacking} problem  is \textsc{W[2]}-hard  for claw-free graphs when parameterized by the solution size.
\end{lemma}


\begin{proof}
    
We reduce the \textsc{Hitting Set} problem to the \textsc{Multipacking} problem  to show that the \textsc{Multipacking} problem  is \textsc{NP}-complete  for claw-free graphs. 

Let $U=\{u_1,u_2,\dots,u_n\}$ and  $\mathcal{F}=\{S_1,S_2,\dots,S_m\}$ where $S_i\subseteq U$ for each $i$. We construct a graph $G$ in the following way: (i)  Include each element of $\mathcal{F}$ in the vertex set of $G$ and $\mathcal{F}$ forms a clique in $G$. (ii) Include $k-3$ length path $u_i = u_i^1 u_i^2 \dots u_i^{k-2}$  in $G$,  for each element $u_i\in U$.
(iii) Include the set of vertices $W=\{w_{j,i}:u_i \notin S_j, 1\leq i\leq n, 1\leq j\leq m\}$ in $G$. Join $S_j$ and $u_i^1$ by a $2$-length path $S_jw_{j,i}u_i^1$. (iv) Join $w_{j,i}$ and $w_{q,p}$ by an edge iff either $j=q$ or $i=p$. Fig. \ref{fig:reduction_claw_free} gives an illustration.

\begin{figure}[ht]
    \centering
   \includegraphics[width=0.9\textwidth]{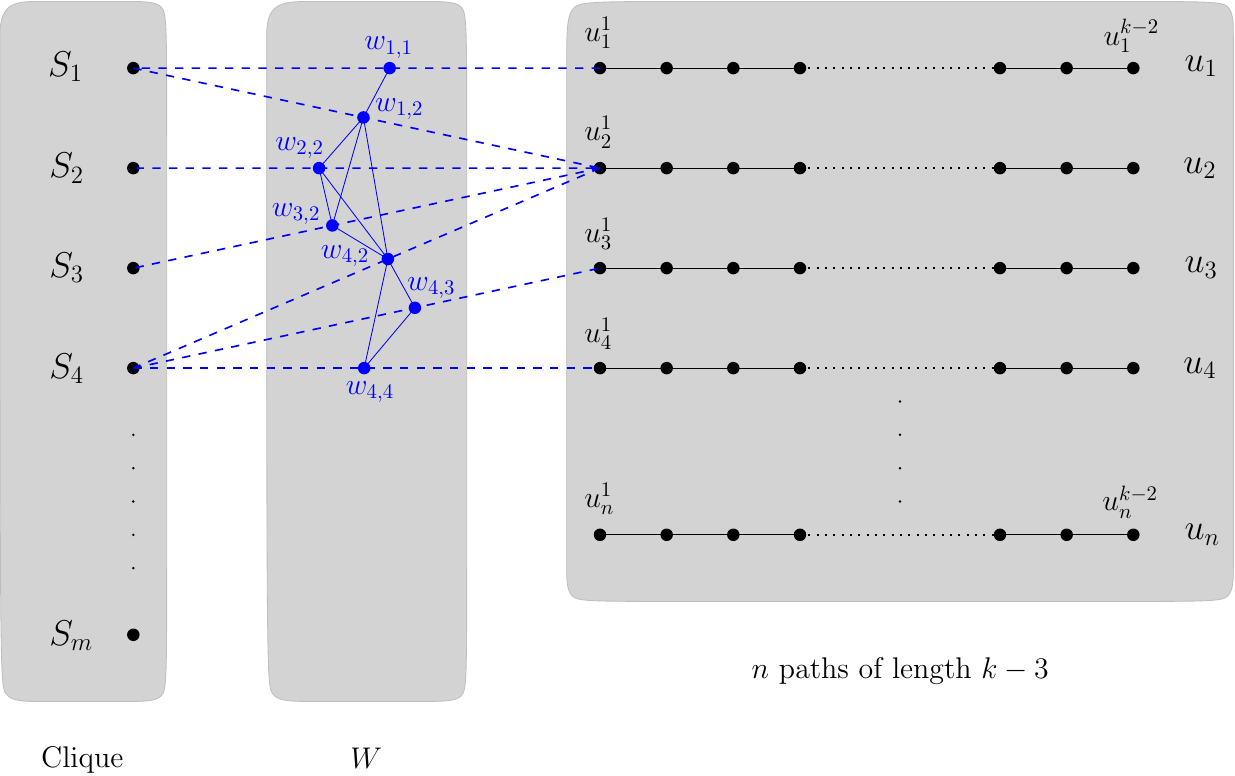}
    \caption{An illustration of the construction of the graph $G$ used in the proof of Lemma~\ref{thm:NPC_clawfree}, demonstrating the hardness of the \textsc{Multipacking} problem for claw-free graphs.}
    \label{fig:reduction_claw_free}
\end{figure}

\vspace{0.3cm}
\noindent
\textbf{Claim \ref{thm:NPC_clawfree}.1.} $G$ is a claw-free graph.

\begin{claimproof} Suppose $G$ contains an induced claw $C$ whose vertex set is $V(C)=\{c,x,y,z\}$ and the claw-center is $c$. Note that $c\notin \{u_i^j:1\leq i\leq n, 2\leq j\leq k-2\}$, since each element of this set has degree at most $2$. 

Suppose $c\in \{u_i^1:1\leq i\leq n\}$. Let $c=u_t^1$ for some $t$. Therefore, $\{x,y,z\}\subseteq N(c)=\{u_t^2\}\cup \{w_{j,t}:u_t\notin S_j,1\leq j\leq m\}$. Then at least $2$ vertices of the set $\{x,y,z\}$ belong to the set $\{w_{j,t}:u_t\notin S_j,1\leq j\leq m\}$. Note that $\{w_{j,t}:u_t\notin S_j,1\leq j\leq m\}$ forms a clique in $G$. This implies that at least $2$ vertices of the set $\{x,y,z\}$ are adjacent. This is a contradiction, since $C$ is an induced claw in $G$. Therefore,  $c\notin \{u_i^1:1\leq i\leq n\}$. 

Suppose $c\in \mathcal{F}$. Let $c=S_t$ for some $t$. Therefore, $\{x,y,z\}\subseteq N(c)=\{S_j:j\neq t, 1\leq j\leq m\}\cup \{w_{t,i}:u_i\notin S_t,1\leq i\leq n\}$. Note that both sets $\{S_j:j\neq t, 1\leq j\leq m\}$ and $ \{w_{t,i}:u_i\notin S_t,1\leq i\leq n\}$ form cliques in $G$. This implies that at least $2$ vertices of the set $\{x,y,z\}$ are adjacent. This is a contradiction, since $C$ is an induced claw in $G$. Therefore, $c\notin \mathcal{F}$.

Suppose $c\in W$. Let $c=w_{q,p}$ for some $q$ and $p$. Let $A_q=\{S_q\}\cup \{w_{q,i}:u_i\notin S_q,i\neq p,1\leq i\leq n\}$ and $B_p=\{u_p^1\}\cup \{w_{j,p}:u_p\notin S_j,j\neq q,1\leq j\leq m\}$. Note that both sets $A_q$ and $B_p$ form cliques in $G$ and $\{x,y,z\}\subseteq N(c)=A_q\cup B_p$. This implies that at least $2$ vertices of the set $\{x,y,z\}$ are adjacent. This is a contradiction, since $C$ is an induced claw in $G$. Therefore, $c\notin W$.

Therefore, $G$ is a claw-free graph. \end{claimproof}

\vspace{0.3cm}
\noindent
\textbf{Claim \ref{thm:NPC_clawfree}.2.} $\mathcal{F}$ has a hitting set of size at most $k$ iff $G$ has a multipacking of size at least $k$, for $k\geq 2$.

\begin{claimproof} \textsf{(If part)} Suppose $\mathcal{F}$ has a hitting set $H = \{ u_1, u_2, \dots, u_k \}$. 
We want to show that $M =\{ u_1^{k-2}, u_2^{k-2}, \dots, u_k^{k-2} \}$ is a multipacking in $G$. If $r < k-1$, then 
$| N_{r}[v] \cap M |\leq 1 \leq r$, $ \forall v \in V(G)$, since $d(u_i^{k-2}, u_j^{k-2})\geq 2k-3$, $\forall i\neq j$.
If $r > k-1$, then 
$| N_{r}[v] \cap M | \leq | M | = k \leq r$, $\forall v \in V(G)$. Suppose $r = k-1$. If $M$ is not a multipacking,  there exists some $v \in V(G)$ such that 
$| N_{k-1}[v] \cap M | = k.$
Therefore, 
$M \subseteq N_{k-1}[v]$.

\vspace{0.22cm}

 \noindent \textit{\textbf{Case 1: }} $v\in \{u_i^j:1\leq i\leq n, 1\leq j\leq k-2\}$.
 
Let $v=u_p^q$ for some $p$ and $q$. In this case, $d(u_p^q,u_i^{k-2})\geq k$, $\forall i\neq p$. Therefore, $v\notin \{u_i^j:1\leq i\leq n, 1\leq j\leq k-2\}$.

\vspace{0.22cm}

 \noindent \textit{\textbf{Case 2: }} $v \in W$.
 
 Let $v=w_{p,q}$ for some $p$ and $q$. Therefore, $u_q\notin S_p$. Since $M \subseteq N_{k-1}[w_{p,q}]$, we have $d(w_{p,q},u_i^{k-2})\leq k-1$, $\forall 1\leq i\leq k$. This implies that $d(w_{p,q},u_i^1)\leq 2$, $\forall 1\leq i\leq k$. Note that $d(w_{p,q},u_i^1)\neq 1$, $\forall i\neq q, 1\leq i\leq k$. Therefore, $d(w_{p,q},u_i^1)=2$, $\forall i\neq q, 1\leq i\leq k$. Suppose $t\neq q$ and $1\leq t\leq k$. Then $d(w_{p,q},u_t^1)=2$.
Therefore, $w_{p,q}$ is adjacent to $w_{h,t}$, for some $h$. This is only possible when $h=p$, since $t\neq q$. Therefore, $u_t^1$ and $S_p$ are joined by a $2$-length path $u_t^1-w_{p,t}-S_p$. This implies that $u_t\notin S_p$. Hence, $u_1, u_2, \dots, u_k \notin S_p$.  Therefore, $H$ is not a hitting set of $\mathcal{F}$. This is a contradiction. Therefore, $v \notin W$.

\vspace{0.22cm}

 \noindent \textit{\textbf{Case 3: }} $v\in \mathcal{F}$.

 Let $v = S_t$ for some $t$. If there exist $i\in\{1,2,\dots,k\}$ such that $u_i\in S_t$, then $d(S_t,u_i^1)\geq 3$. Therefore, $d(S_t,u_i^{k-2})\geq k$. This implies that $u_i^{k-2}\notin N_{k-1}[S_t]$. Therefore, $M \nsubseteq N_{k-1}[S_t] $. Hence, $M \subseteq N_{k-1}[S_t] $ implies that  $u_1, u_2, \dots, u_k \notin S_t$. Then $H$ is not a hitting set of $\mathcal{F}$. This is a contradiction. Therefore, $v\notin \mathcal{F}$.

Hence, $M$ is a multipacking of size $k$.

\medskip
 \textsf{(Only-if part)}  Suppose $G$ has a multipacking $M$ of size $k$. Let 
$H = \{ u_i : u_i^j \in M, 1\leq i\leq n, 1\leq j\leq k-1 \}.$
Then 
$|H| \leq |M| = k.$
Without loss of generality, let 
$H = \{ u_1, u_2, \dots, u_{k'} \} $ where $ k' \leq k$.
We want to show that $H$ is a hitting set of $\mathcal{F}$. Suppose $H$ is not a hitting set. Then there exists $t$ such that 
$S_t \cap H = \emptyset.$
Therefore,  $S_t$ and $u_i^1$ are joined by a $2$-length path $S_t-w_{t,i}-u_i^1$, for each $i\in\{1,2,\dots,k'\}$. 
So, $M \subseteq N_{k-1}[S_t]$. This implies that 
$| N_{k-1}[S_t] \cap M | = k.$
This is a contradiction. Therefore, $H$ is a hitting set of size at most $k$. 
\end{claimproof}

\noindent Since the above reduction is a polynomial-time, the \textsc{Multipacking} problem  is \textsc{NP}-complete  for claw-free graphs.  Moreover, the reduction is also an FPT reduction. Hence, the \textsc{Multipacking} problem  is \textsc{W[2]}-hard  for claw-free graphs when parameterized by the solution size.
\end{proof}

\subsection{Regular graphs}

Next, we prove NP-completeness for regular graphs. We need the following result to prove our theorem.




From the Erdős-Gallai Theorem~\cite{erdos1960grafok} and the Havel-Hakimi criterion~\cite{hakimi1962realizability,havel1955remark}, we can say the following.

\begin{lemma}[\cite{erdos1960grafok,hakimi1962realizability,havel1955remark}]\label{lem:ErdosHavelHakimi}
   Let $n$ and $d$ be nonnegative integers. If $d<n$ and $nd$ is even, then there exists a simple $d$-regular graph on $n$ vertices. Moreover, such a graph can be constructed in polynomial time.
\end{lemma}


We use Lemma \ref{lem:ErdosHavelHakimi} to prove the NP-completeness of the \textsc{Multipacking} problem on regular graphs.


\begin{restatable}{lemma}{multipackingregularNPc}\label{thm:multipacking_regular_NPc}
    \textsc{Multipacking} problem  is \textsc{NP}-complete  for regular graphs. 
\end{restatable}

\begin{proof} It is known that the \textsc{Total Dominating Set} problem is \textsc{NP}-complete  for cubic (3-regular) graphs~\cite{garey1979computers}.  We reduce the \textsc{Total Dominating Set} problem on cubic graphs to the \textsc{Multipacking} problem  on regular graphs. 

\medskip
\noindent
\fbox{%
  \begin{minipage}{\dimexpr\linewidth-2\fboxsep-2\fboxrule}
  \textsc{ Total Dominating Set} problem
  \begin{itemize}[leftmargin=1.5em]
    \item \textbf{Input:} An undirected graph $G = (V, E)$ and an integer $k \in \mathbb{N}$.
    \item \textbf{Question:} Does there exist a \emph{total dominating set} $S \subseteq V$ of size at most $k$; that is, a set of at most $k$ vertices such that every vertex in $V$ has at least one neighbor in $S$?
  \end{itemize}
  \end{minipage}%
}
\medskip

Let $(G,k)$ be an instance of the \textsc{Total Dominating Set} problem, where $G$ is a $3$-regular graph  with the vertex set $V=\{v_1,v_2,\dots,v_n\}$ where $n\geq 6$. Therefore, $n$ is even. Let $d=n-4$. So, $d$ is also even. Now we construct a graph $G'$ in the following way (Fig. \ref{fig:reduction_regular} gives an illustration). 

\noindent(i) For every vertex $v_a\in V(G)$, we construct a graph $H_a$ in $G'$ as follows. 

\begin{itemize}
    \item Add the vertex sets $S_a^i=\{u_a^{i,j}:1\leq j\leq d\}$ for all $2\leq i\leq k-2$ in $H_a$. Each vertex of $S_a^i$ is adjacent to each vertex of $S_a^{i+1}$ for each $i$.

    \item Add a vertex $u_a^{1}$ in $H_a$ where $u_a^{1}$ is adjacent to each vertex of $S_a^2$. Moreover, $S_a^2$ forms a clique.

    \item  Add the vertex set $T_a=\{u_a^{k-1,j}:1\leq j\leq d^2\}$ in $H_a$ where the induced subgraph $H_a[T_a]$ is a $(2d-1)$-regular graph. We can construct such a graph in polynomial-time by Lemma \ref{lem:ErdosHavelHakimi} since $d$ is an even number and $2d-1 < d^2 $ (since $n\geq 6$). 

    \item Let $U_1, U_2, \ldots, U_d$ be a partition of $T_a$ into $d$ disjoint sets, each of size $d$, that is, $T_a=U_1\sqcup U_2\sqcup \dots \sqcup U_d$ and $|U_j|=d$ for each $1\leq j\leq d$. Each vertex in $U_j$ is adjacent to $u_a^{k-2,j}$ for all $1\leq j\leq d$.
    
\end{itemize}

\noindent(ii) If $v_i$ and $v_j$ are different and not adjacent in $G$, we make $u^1_i$ and $u^1_j$ adjacent in $G'$ for every $1\leq i,j\leq n, i\neq j$. \\

Note that $G'$ is a $2d$-regular graph. 

\begin{figure}[ht]
    \centering
   \includegraphics[width=0.95\textwidth]{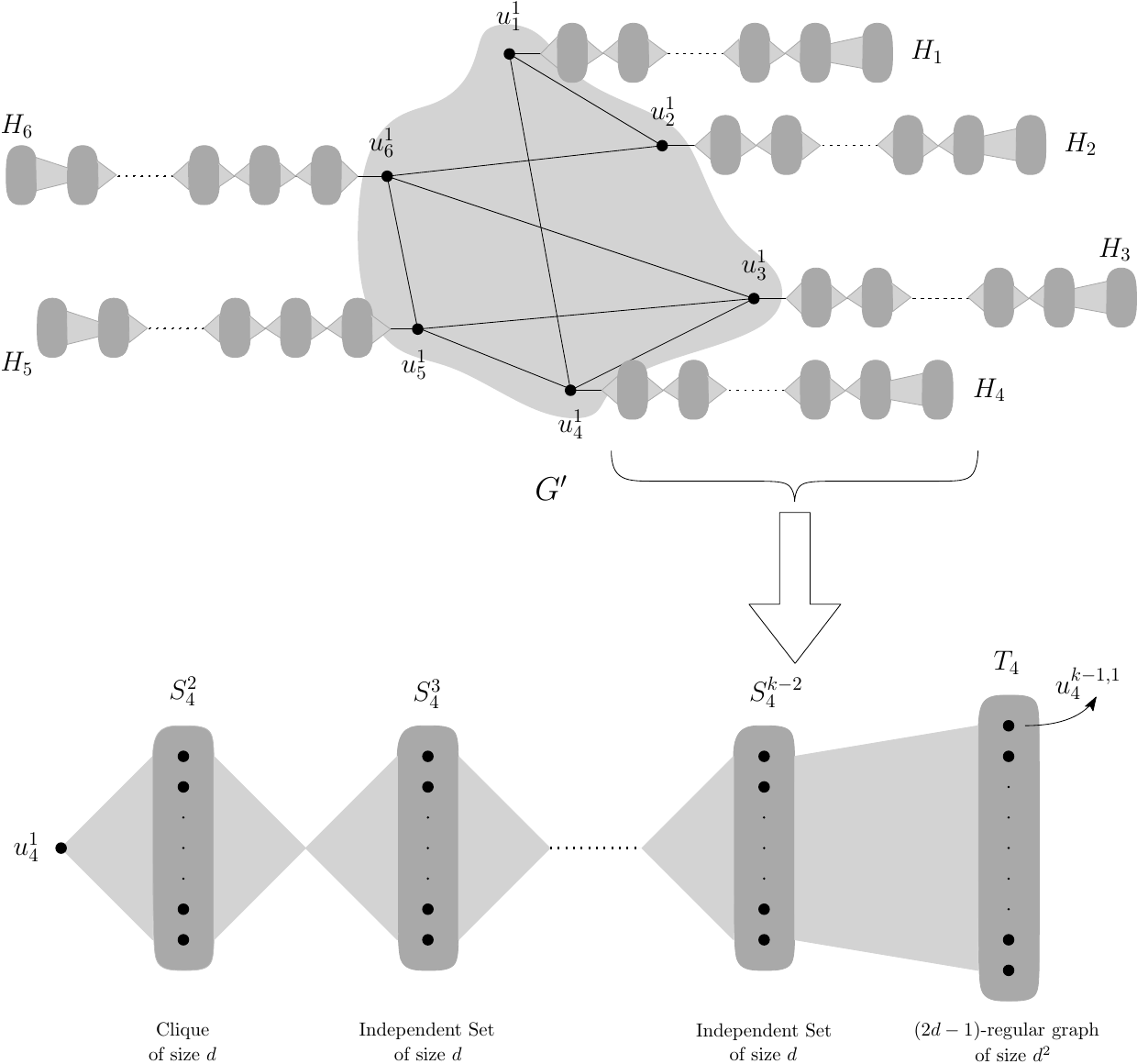}
    \caption{An illustration of the construction of the graph $G'$ used in the proof of Lemma~\ref{thm:multipacking_regular_NPc}, demonstrating the hardness of the \textsc{Multipacking} problem for regular graphs. The induced subgraph $G'[u_1^1,u_2^1,\dots,u_6^1]$ is isomorphic to $G^c$.}
    \label{fig:reduction_regular}
\end{figure}

\vspace{0.3cm}
\noindent
\textbf{Claim \ref{thm:multipacking_regular_NPc}.1. } $G$ has a total dominating set of size at most $k$ iff $G'$ has a multipacking of size at least $k$, for $k\geq 2$.

\begin{claimproof} \textsf{(If part)} 
Suppose $G$ has a total dominating set $D = \{ v_1, v_2, \dots, v_k \}$. 
We want to show that $M =\{ u_1^{k-1,1}, u_2^{k-1,1}, \dots, u_k^{k-1,1} \}$ is a multipacking in $G'$.  If $r < k-1$, then 
$| N_{r}[v] \cap M |\leq 1 \leq r$, $ \forall v \in V(G')$, since $d(u_i^{k-1,1}, u_j^{k-1,1})\geq 2k-3$, $\forall i\neq j$.
If $r > k-1$, then 
$| N_{r}[v] \cap M | \leq | M | = k \leq r$, $\forall v \in V(G')$. Suppose $r = k-1$. If $M$ is not a multipacking,  there exists some vertex $v \in V(G')$ such that 
$| N_{k-1}[v] \cap M | = k$. Therefore, 
$M \subseteq N_{k-1}[v]$. Suppose $v\in V(H_p)\setminus\{u_p^1\}$ for some $p$. In this case, $d(v,u_i^{k-1,1})\geq k$, $\forall i\neq p$. Therefore, $v\notin V(H_p)\setminus\{u_p^1\}$ for any $p$.  This implies that $v \in \{u_i^1:1\leq i\leq n\}$. Let $v = u_t^1$ for some $t$. If there exist $i\in\{1,2,\dots,k\}$ such that neither $u_t^1=u_i^1$ nor $u_t^1$ is adjacent to $u_i^1$, then $d(u_t^1,u_i^{k-1,1})\geq k$. Therefore, 
$M \subseteq N_{k-1}[u_t^1] $ implies that either  $u_t^1$ is adjacent to each vertex of the set  $\{u_1^1, u_2^1, \dots, u_k^1\}$ when $u_t^1\notin \{u_1^1, u_2^1, \dots, u_k^1\}$ or $u_t^1$ is adjacent to each vertex of the set  $\{u_1^1, u_2^1, \dots, u_k^1\}\setminus \{u_t^1\}$ when $u_t^1\in \{u_1^1, u_2^1, \dots, u_k^1\}$. Therefore, in the graph $G$, either $v_t$ is not adjacent to any vertex of $D$ when $v_t\notin D$ or  $v_t$ is not adjacent to any vertex of $D\setminus\{v_t\}$ when $v_t\in D$. In both cases, $D$ is not a total dominating set of $G$. This is a contradiction.    Hence, $M$ is a multipacking of size $k$ in $G'$.

 \textsf{(Only-if part)}  Suppose $G'$ has a multipacking $M$ of size $k$. Let 
$D = \{ v_i : V(H_i)\cap M\neq \emptyset, 1\leq i\leq n \}.$
Then 
$|D| \leq |M| = k.$
Without loss of generality, let 
$D = \{ v_1, v_2, \dots, v_{k'} \} $ where $ k' \leq k$.
Now we show that $D$ is a total dominating set of $G$. Suppose $D$ is not a total dominating set.  Then there exists $t$ such that either 
$v_t\notin D$ and no vertex in $D$ is adjacent to $v_t$ or $v_t\in D$ and no vertex in $D\setminus\{v_t\}$ is adjacent to $v_t$. Therefore, in $G'$, either  $u_t^1$ is adjacent to each vertex of the set  $\{u_1^1, u_2^1, \dots, u_k^1\}$ when $u_t^1\notin \{u_1^1, u_2^1, \dots, u_k^1\}$ or $u_t^1$ is adjacent to each vertex of the set  $\{u_1^1, u_2^1, \dots, u_k^1\}\setminus \{u_t^1\}$ when $u_t^1\in \{u_1^1, u_2^1, \dots, u_k^1\}$.  Then $M \subseteq N_{k-1}[u_t^1]$. This implies that 
$| N_{k-1}[u_t^1] \cap M | = k$.
This is a contradiction. Therefore, $D$ is a total dominating set of size at most $k$ in $G$.   
\end{claimproof}

\noindent Hence, the \textsc{Multipacking} problem  is \textsc{NP}-complete  for regular graphs. 
\end{proof}

\subsection{Convex Intersection graphs}

Here, we discuss the hardness of the \textsc{Multipacking} problem for a geometric intersection graph class: CONV graphs.  A graph $G$ is a \textit{CONV graph} if and only if there exists  a family of convex sets on a plane such that the graph has a vertex for each convex set and an edge for each intersecting pair of convex sets. We show that the \textsc{Multipacking} problem is \textsc{NP}-complete  for CONV graphs. To prove this we will need the following result.

\begin{theorem}[\cite{kratochvil1998intersection}]\label{thm:coplanar_CONV}
    The complement of a planar graph (also called co-planar graph) is a CONV graph. 
\end{theorem}

We use Theorem \ref{thm:coplanar_CONV} to prove the NP-completeness of the \textsc{Multipacking} problem on CONV graphs.

    

\begin{lemma}\label{thm:multipacking_CONV_NPc}
    \textsc{Multipacking} problem  is \textsc{NP}-complete  for CONV graphs. 
\end{lemma}


\begin{proof} It is known that the \textsc{Total Dominating Set} problem is \textsc{NP}-complete  on planar graphs~\cite{garey1979computers}. We reduce the \textsc{Total Dominating Set} problem on planar graphs to the \textsc{Multipacking} problem  on CONV graphs to show that the \textsc{Multipacking} problem  is \textsc{NP}-complete  on CONV graphs.

Let $(G,k)$ be an instance of the \textsc{Total Dominating Set}  problem, where $G$ is a planar graph  with the vertex set $V=\{v_1,v_2,\dots,v_n\}$. We construct a graph $G'$ in the following way (Fig. \ref{fig:reduction_CONV} gives an illustration). 

\noindent(i) For every vertex $v_i\in V(G)$, we introduce a path $u_i = u_i^1 u_i^2 \dots u_i^{k-1}$ of length $k-2$ in $G'$. 

\noindent(ii) If $v_i$ and $v_j$ are different and not adjacent in $G$, we make $u^1_i$ and $u^1_j$ adjacent in $G'$ for every $1\leq i,j\leq n, i\neq j$.

\begin{figure}[ht]
    \centering
   \includegraphics[width=0.9\textwidth]{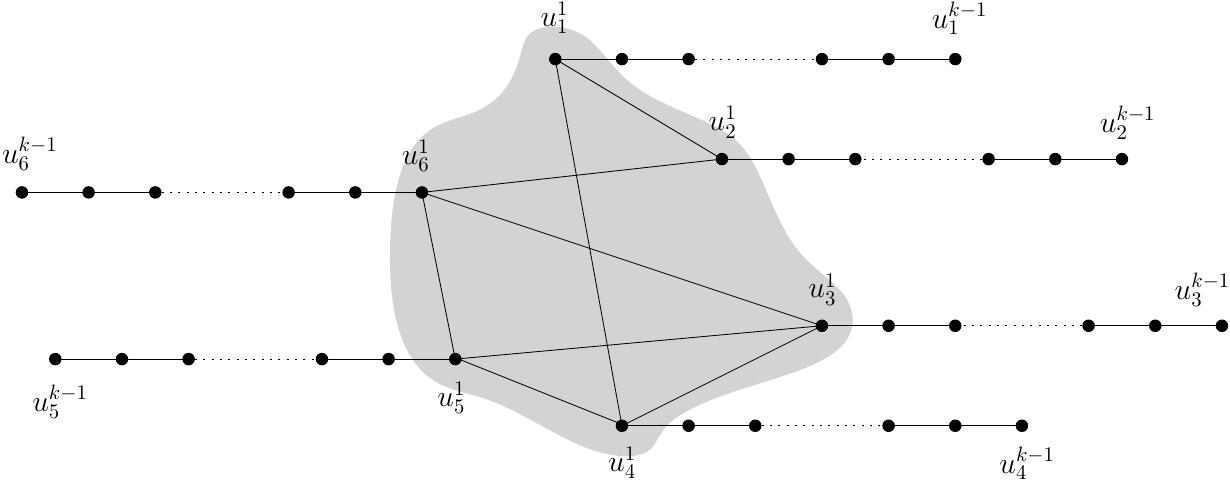}
    \caption{An illustration of the construction of the graph $G'$ used in the proof of Lemma~\ref{thm:multipacking_CONV_NPc}, demonstrating the hardness of the \textsc{Multipacking} problem for CONV graphs. The induced subgraph $G'[u_1^1,u_2^1,\dots,u_6^1]$ (the colored region) is isomorphic to $G^c$.}
    \label{fig:reduction_CONV}
\end{figure}

\vspace{0.3cm}
\noindent
\textbf{Claim \ref{thm:multipacking_CONV_NPc}.1. } $G'$ is a CONV graph.

\begin{claimproof} Note that the induced subgraph $G'[u_1^1,u_2^1,\dots,u_n^1]$ is isomorphic to $G^c$. Since $G$ is a planar graph, $G'[u_1^1,u_2^1,\dots,u_n^1]$ is a CONV graph by Lemma \ref{thm:coplanar_CONV}. Hence, $G'$ is a CONV graph.    
\end{claimproof}

\vspace{0.3cm}
\noindent
\textbf{Claim \ref{thm:multipacking_CONV_NPc}.2. } $G$ has a total dominating set of size at most $k$ iff $G'$ has a multipacking of size at least $k$, for $k\geq 2$.

\begin{claimproof} \textsf{(If part)} 
Suppose $G$ has a total dominating set $D = \{ v_1, v_2, \dots, v_k \}$. 
We want to show that $M =\{ u_1^{k-1}, u_2^{k-1}, \dots, u_k^{k-1} \}$ is a multipacking in $G'$.  If $r < k-1$, then 
$| N_{r}[v] \cap M |\leq 1 \leq r$, $ \forall v \in V(G')$, since $d(u_i^{k-1}, u_j^{k-1})\geq 2k-3$, $\forall i\neq j$.
If $r > k-1$, then 
$| N_{r}[v] \cap M | \leq | M | = k \leq r$, $\forall v \in V(G')$. Suppose $r = k-1$. If $M$ is not a multipacking,  there exists some $v \in V(G')$ such that 
$| N_{k-1}[v] \cap M | = k$. Therefore, 
$M \subseteq N_{k-1}[v].$ Suppose $v=u_p^q$ for some $p$ and $q$ where $1\leq p\leq n, 2\leq q\leq k-1$. In this case, $d(u_p^q,u_i^{k-1})\geq k$, $\forall i\neq p$. Therefore, $v\notin \{u_i^j:1\leq i\leq n, 2\leq j\leq k-1\}$.  This implies that $v \in \{u_i^1:1\leq i\leq n\}$.  Let $v = u_t^1$ for some $t$. If there exist $i\in\{1,2,\dots,k\}$ such that neither $u_t^1=u_i^1$ nor $u_t^1$ is adjacent to $u_i^1$, then $d(u_t^1,u_i^{k-1})\geq k$. Therefore, 
$M \subseteq N_{k-1}[u_t^1] $ implies that either  $u_t^1$ is adjacent to each vertex of the set  $\{u_1^1, u_2^1, \dots, u_k^1\}$ when $u_t^1\notin \{u_1^1, u_2^1, \dots, u_k^1\}$ or $u_t^1$ is adjacent to each vertex of the set  $\{u_1^1, u_2^1, \dots, u_k^1\}\setminus \{u_t^1\}$ when $u_t^1\in \{u_1^1, u_2^1, \dots, u_k^1\}$. Therefore, in the graph $G$, either $v_t$ is not adjacent to any vertex of $D$ when $v_t\notin D$ or  $v_t$ is not adjacent to any vertex of $D\setminus\{v_t\}$ when $v_t\in D$. In both cases, $D$ is not a total dominating set of $G$. This is a contradiction.    Hence, $M$ is a multipacking of size $k$ in $G'$.

 \textsf{(Only-if part)}  Suppose $G'$ has a multipacking $M$ of size $k$. Let 
$D = \{ v_i : u_i^j \in M, 1\leq i\leq n, 1\leq j\leq k-1 \}.$
Then 
$|D| \leq |M| = k.$
Without loss of generality, let 
$D = \{ v_1, v_2, \dots, v_{k'} \} $ where $ k' \leq k$.
Now we show that $D$ is a total dominating set of $G$. Suppose $D$ is not a total dominating set.  Then there exists $t$ such that either 
$v_t\notin D$ and no vertex in $D$ is adjacent to $v_t$ or $v_t\in D$ and no vertex in $D\setminus\{v_t\}$ is adjacent to $v_t$. Therefore, in $G'$, either  $u_t^1$ is adjacent to each vertex of the set  $\{u_1^1, u_2^1, \dots, u_k^1\}$ when $u_t^1\notin \{u_1^1, u_2^1, \dots, u_k^1\}$ or $u_t^1$ is adjacent to each vertex of the set  $\{u_1^1, u_2^1, \dots, u_k^1\}\setminus \{u_t^1\}$ when $u_t^1\in \{u_1^1, u_2^1, \dots, u_k^1\}$.  Then $M \subseteq N_{k-1}[u_t^1]$. This implies that 
$| N_{k-1}[u_t^1] \cap M | = k$.
This is a contradiction. Therefore, $D$ is a total dominating set of size at most $k$ in $G$.   
\end{claimproof}

\noindent Hence, the \textsc{Multipacking} problem  is \textsc{NP}-complete  for CONV graphs. 
\end{proof}

\section{Exact Exponential Algorithm (Proof of Theorem~\ref{thm:1.53time})}\label{sec:exponentialalgorithm}

In this section, we present an exact exponential-time algorithm for computing
a maximum multipacking in a graph. As usual, the $O^*$-notation suppresses polynomial factors, i.e. $O^*(c^n)$ denotes a running time of the form $O(c^n\cdot p(n))$, where $p(n)$ is a polynomial in $n$.

Let $\mathcal{M}(G)$ denote the family of all multipackings of a graph $G$.
Recall that a set $M\subseteq V(G)$ is a multipacking if
$|N_r[z]\cap M|\leq r$ for every vertex $z\in V(G)$ and every integer $r\geq1$,
where $N_r[z]=\{u\in V(G):d_G(z,u)\leq r\}$.
Suppose that $G$ is connected and $T$ is a spanning tree of $G$.
Since $d_T(u,z)\geq d_G(u,z)$ for any $u,z\in V(G)$, we have
$\mathcal{M}(G)\subseteq\mathcal{M}(T)$.
Our algorithm computes a family of subsets of $V(T)$ that contains
$\mathcal{M}(T)$, and then checks these subsets in the original graph $G$ to obtain a maximum multipacking of $G$.

We first introduce some notation. Let $T$ be a tree rooted at $v$.
For a vertex $x\in V(T)$, let $T_x$ denote the subtree consisting of $x$ and
all its descendants, rooted at $x$.
The height of $T$, denoted by $\operatorname{height}(T)$, is the maximum
distance from $v$ to a vertex of $T$.
We write $T_x\cong H$ when the two rooted trees are isomorphic by a map
that preserves the roots, i.e., there exists a bijection
$f(T_x)\to V(H)$
such that $uv\in E(T_x)$ if and only if $f(u)f(v)\in E(H)$ for every $u,v\in V(T_x)$, and moreover $f(x)$ is the root of $H$. For a vertex set $X$, let $T-X:=T[V(T)\setminus X]$, that is, the subgraph of $G$ induced by the vertex set $V(T)\setminus X$. 
For a subtree $T_x$, we write $T\setminus T_x:=T-V(T_x)$.
For two families of sets $\mathcal{S}_1$ and $\mathcal{S}_2$, define
$$
\mathcal{S}_1\oplus\mathcal{S}_2
:=\{M_1\cup M_2:M_1\in\mathcal{S}_1,\ M_2\in\mathcal{S}_2\}.
$$
For a vertex $u$, we write
$u\oplus\mathcal{S}:=\{\{u\}\}\oplus\mathcal{S}$.

First, we present a simple exact exponential-time algorithm for computing a maximum multipacking that runs in time $O(1.62^n)$. Subsequently, we describe an improved exponential-time algorithm with a better running time.

\begin{proposition}\label{thm:1.62time}
   There is an exact exponential-time algorithm for computing a maximum multipacking in a graph with $n$ vertices, which has running time $O(1.62^n)$.

\end{proposition}

\begin{proof}
    Let $G$ be a connected graph with $n$ vertices and $T$ be a BFS (Breadth-First Search) tree of $G$  rooted at $v$. Note that every multipacking in $G$ is a multipacking in $T$. For $n=1$, computing a maximum multipacking is trivial. Assume $n\geq 2$. Let $w$ be a farthest vertex from $v$ and $y$ be the parent of $w$. Let $T_{y}$ be the subtree of $T$ rooted at $y$.  If a multipacking $M_1$ of $T$ contains $w$, then no other vertices of $T_{y}$ belong to $M_1$; otherwise two vertices of $M_1$ would belong to $N_1[y]$. Therefore, $M_1\in w\oplus \mathcal{M}(T\setminus T_y) $. If a multipacking $M_2$ of $T$ does not contain $w$, then $M_2\in \mathcal{M}(T\setminus w)$ where $T\setminus w$ is the subtree of $T$ formed by deleting $w$ from $T$. Therefore, $\mathcal{M}(T)\subseteq (w\oplus \mathcal{M}(T\setminus T_y))\cup \mathcal{M}(T\setminus w)$. This implies the recursive relation $|\mathcal{M}(T)|\leq | \mathcal{M}(T\setminus T_y)|+ |\mathcal{M}(T\setminus w)|$, since $|w\oplus \mathcal{M}(T\setminus T_y)|= |\mathcal{M}(T\setminus T_y)|$. Let $f(n)$ be the maximum number of multipackings in a tree on $n$ vertices. So, the recursive relation yields that $f(n)\leq f(n-2)+f(n-1)$, which is the Fibonacci recurrence. Solving
this recurrence yields $f(n)=O(1.6181^n)$. Therefore, we can compute a family of subsets of $V(G)$ having size $O(1.6181^n)$ that contains all the multipackings of $G$. By searching within this family, we can find a maximum multipacking of $G$ in $O^*(1.6181^n)$  time. Note that, $O^*(1.6181^n)=O(1.62^n)$.\footnote{For sufficiently large $n$, we have $1.6181^nn^{O(1)}< 1.62^n$. Hence $O^*(1.6181^n)$ can be replaced by $O(1.62^n)$.}
\end{proof}

Now we refine the above branching idea to obtain a better running time. Let $k_1,k_2$ be nonnegative integers. Let $H(k_1,k_2)$ denote the rooted
tree with root $x$, vertex set
$$
\{x,u_1,\ldots,u_{k_1},a_1,\ldots,a_{k_2},b_1,\ldots,b_{k_2}\},
$$
and edge set
$$
\{xu_i:1\leq i\leq k_1\}
\cup\{xa_i:1\leq i\leq k_2\}
\cup\{a_ib_i:1\leq i\leq k_2\}.
$$
Thus, $H(k_1,k_2)$ consists of $k_1$ paths of length $1$ and $k_2$ paths
of length $2$, all having the common vertex $x$; see Fig.~\ref{fig:H}.
In particular, $|V(H(k_1,k_2))|=1+k_1+2k_2$.
A vertex set indexed by an empty index range is understood to be empty.

\begin{figure}[ht]
    \centering
   \includegraphics[width=0.5\textwidth]{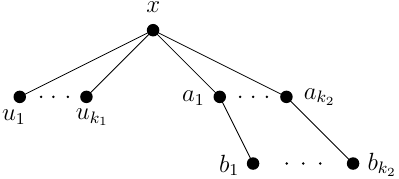}
\caption{The rooted tree $H(k_1,k_2)$.}
\label{fig:H}
\end{figure}

\begin{lemma}\label{lem:Hkpolybounds}
Let $m=|V(H(k_1,k_2))|$. Then
$$
|\mathcal{M}(H(k_1,k_2))|
=1+m+k_1k_2+\binom{k_2}{2}+k_2(k_2-1).
$$
Moreover, if $m\geq4$, then
$$
|\mathcal{M}(H(k_1,k_2))|
\leq \frac38m^2-\frac12m+\frac{17}{8}.
$$
The family $\mathcal{M}(H(k_1,k_2))$ can be computed in time
$O(|\mathcal{M}(H(k_1,k_2))|)$.
\end{lemma}

\begin{proof}
Let $U=\{u_1,\ldots,u_{k_1}\}$, $A=\{a_1,\ldots,a_{k_2}\}$, and
$B=\{b_1,\ldots,b_{k_2}\}$. Since every vertex of $H(k_1,k_2)$ is within
distance $2$ of $x$, every multipacking has size at most $2$.
Let $\mathcal{M}_0,\mathcal{M}_1,\mathcal{M}_2$ denote the families of its multipackings of sizes $0,1,2$,
respectively. Clearly, $\mathcal{M}_0=\{\emptyset\}$ and
$\mathcal{M}_1=\{\{z\}:z\in V(H(k_1,k_2))\}$. Thus $|\mathcal{M}_0|=1$ and $|\mathcal{M}_1|=m$.

Note that, a set of two vertices is a multipacking if and only if the vertices are
at distance at least $3$.
Consequently,
$$
\begin{aligned}
\mathcal{M}_2={}&
\{\{u_i,b_j\}:1\leq i\leq k_1,\ 1\leq j\leq k_2\}\\
&\cup\{\{b_i,b_j\}:1\leq i<j\leq k_2\}\\
&\cup\{\{a_i,b_j\}:1\leq i,j\leq k_2,\ i\neq j\}.
\end{aligned}
$$

These three families are disjoint. Therefore,
$|\mathcal{M}_2|=k_1k_2+\binom{k_2}{2}+k_2(k_2-1)$, which proves the equality.

Now suppose $m\geq4$. Since $k_1=m-2k_2-1$, we have
$$ 0\leq k_2\leq(m-1)/2
 \quad\text{and}\quad|\mathcal{M}_2|=\frac{2m-5}{2}\,k_2-\frac12k_2^2.
$$
Let $g(t)=\frac{2m-5}{2}t-\frac12t^2$.
For $0\leq t\leq(m-1)/2$, we have
$g'(t)=(2m-5)/2-t\geq(m-4)/2\geq0$, i.e. $g(t)$ is a non-decreasing function in the interval $[0,(m-1)/2]$.
Hence
$$\begin{aligned}
|\mathcal{M}(H(k_1,k_2))|
&=1+m+g(k_2)\\
&\leq1+m+g\left(\frac{m-1}{2}\right)
=\frac38m^2-\frac12m+\frac{17}{8}.
\end{aligned}$$

From the constructions of  $\mathcal{M}_0,\mathcal{M}_1$ and $\mathcal{M}_2$,  the family $\mathcal{M}(H(k_1,k_2))$ can be computed in time $O(|\mathcal{M}(H(k_1,k_2))|)$.
\end{proof}

\begin{corollary}\label{cor:H15283}
Let $k_2\geq1$ and
$(k_1,k_2)\notin\{(0,1),(1,1),(0,2)\}$.
If $m=|V(H(k_1,k_2))|$, then
$|\mathcal{M}(H(k_1,k_2))|\leq1.5283^m$.
Moreover, the family can be computed in time $O(1.5283^m)$.
\end{corollary}

\begin{proof}
Let $\beta=1.5283$ and
$P(m)=\frac38m^2-\frac12m+\frac{17}{8}$.
The assumptions $k_2\geq1$ and
$(k_1,k_2)\notin\{(0,1),(1,1),(0,2)\}$ imply $m\geq5$.
If $m=5$, the only allowed pair is $(k_1,k_2)=(2,1)$.
By Lemma~\ref{lem:Hkpolybounds},
$|\mathcal{M}(H(2,1))|=8<\beta^5$.

Suppose $m\geq6$. We have $P(6)=101/8<\beta^6$.
For every integer $m\geq6$,
$$
\frac32P(m)-P(m+1)
=\frac{3m^2-16m+19}{16}
=\frac{3(m-6)^2+20(m-6)+31}{16}>0.
$$
It follows by induction that
$$
P(m)\leq P(6)\left(\frac32\right)^{m-6}
<\beta^6\beta^{m-6}=\beta^m,
$$
since $3/2<\beta$. The result now follows from
Lemma~\ref{lem:Hkpolybounds}.
\end{proof}

We next describe the  algorithm.  When $T$ has height at most $1$, every multipacking has size at most $1$,
so $\mathcal{M}(T)=\{\emptyset\}\cup\{\{z\}:z\in V(T)\}$.
Otherwise, let $S_v$ be the set of vertices farthest from the root $v$.
For $w\in S_v$, let $w_1$ and $w_2$ denote its parent and grandparent,
respectively. All children of $w_1$ are leaves. The first branching case
applies when some such $w_1$ has at least two children.
If this case does not apply, then every subtree $T_{w_2}$ has the form
$H(k_1,k_2)$ with $k_2\geq1$. We handle $H(0,1)$, $H(1,1)$, and $H(0,2)$
separately, and use Corollary~\ref{cor:H15283} for the remaining subtrees. The notation $\mathcal{A}_v(T)$ will denote the family returned by our
 algorithm.

\begin{algorithm}[ht]
\caption{$\mathcal{A}_v(T)$: a family containing all multipackings of a tree $T$ rooted at $v$}
\label{alg:multipackingsuperset}
\small
\begin{algorithmic}[1]
\State \textbf{Input:} A tree $T$ rooted at $v$. 
\State \textbf{Output:} A family $\mathcal{A}_v(T)$ containing $\mathcal{M}(T)$.
\If{$T=\emptyset$}
    \State \Return $\{\emptyset\}$.
\EndIf
\If{$\operatorname{height}(T)\leq1$}
    \State \Return $\{\emptyset\}\cup\{\{z\}:z\in V(T)\}$.
\EndIf
\State Let $S_v$ be the set of vertices farthest from $v$.
\State For each $w\in S_v$, let $w_1,w_2$ be its parent and grandparent respectively.
\If{there exists $w\in S_v$ such that $w_1$ has at least two children}
    \State \Return $(w\oplus\mathcal{A}_v(T\setminus T_{w_1}))\cup\mathcal{A}_v(T-\{w\})$.
\EndIf
\If{there exists $w\in S_v$ such that $T_{w_2}\cong H(0,1)$}
    \State \Return $(w\oplus\mathcal{A}_v(T\setminus T_{w_2}))\cup\mathcal{A}_v(T-\{w\})$.
\EndIf
\If{there exists $w\in S_v$ such that $T_{w_2}\cong H(1,1)$}
    \State Let $x=w_2$, and label $T_x$ as in Fig.~\ref{fig:H}.
    \State $\mathcal{S}\gets\{\{u_1\},\{b_1\},\{u_1,b_1\}\}$.
    \State \Return $\mathcal{A}_v(T-\{u_1,b_1\})\cup(\mathcal{S}\oplus\mathcal{A}_v(T\setminus T_x))$.
\EndIf
\If{there exists $w\in S_v$ such that $T_{w_2}\cong H(0,2)$}
    \State Let $x=w_2$, and label $T_x$ as in Fig.~\ref{fig:H}.
    \State $\mathcal{S}\gets\{\{b_2\},\{a_1\},\{a_1,b_2\},\{b_1\},\{b_1,a_2\},\{b_1,b_2\}\}$.
    \State \Return $\mathcal{A}_v(T-\{a_1,b_1,b_2\})\cup(\mathcal{S}\oplus\mathcal{A}_v(T\setminus T_x))$.
\EndIf
\State Choose any $w\in S_v$ and let $x=w_2$.
\State \Return $\mathcal{A}_v(T\setminus T_x)\oplus\mathcal{M}(T_x)$.
\end{algorithmic}
\end{algorithm}

\begin{lemma}\label{lem:15283setbound}
Let $T$ be an $n$-vertex rooted tree with root $v$.
Algorithm~\ref{alg:multipackingsuperset} computes a family
$\mathcal{A}_v(T)$ of size $O(1.5283^n)$ that contains all multipackings of $T$,
and its running time is $O^*(1.5283^n)$.
\end{lemma}

\begin{proof}
We first prove correctness by induction on $|V(T)|$.
The empty instance and the case $\operatorname{height}(T)\leq1$ are immediate.
Assume $\operatorname{height}(T)\geq2$, and let $M\in\mathcal{M}(T)$.
Every nonempty recursive instance in the algorithm is a connected subtree
containing $v$. Moreover, a connected subtree of a tree is isometric.
Thus the restriction of $M$ to any of these subtrees is a multipacking
of that subtree.

\medskip
\noindent\textit{\textbf{Case 1:} Some $w\in S_v$ has a parent $w_1$ with at least two children, i.e., $T_{w_1}\cong H(k_1,0)$ for some $k_1\geq 2$.}

\smallskip
\noindent All children of $w_1$ are leaves, since $w$ is farthest from $v$.
If $w\in M$, then no other vertex of $T_{w_1}$ belongs to $M$, because
$V(T_{w_1})\subseteq N_1[w_1]$. Therefore,
$M\in w\oplus\mathcal{A}_v(T\setminus T_{w_1})$.
If $w\notin M$, then $M\in\mathcal{A}_v(T-\{w\})$.

\medskip
\noindent\textit{\textbf{Case 2:} $T_{w_2}\cong H(0,1)$ for some $w\in S_v$.}

\smallskip
\noindent The subtree $T_{w_2}$ is the path $w_2w_1w$.
If $w\in M$, neither $w_1$ nor $w_2$ can belong to $M$, since any two
vertices of a multipacking are at distance at least $3$.
Thus $M\in w\oplus\mathcal{A}_v(T\setminus T_{w_2})$.
If $w\notin M$, then $M\in\mathcal{A}_v(T-\{w\})$.

\medskip
\noindent\textit{\textbf{Case 3:} $T_x\cong H(1,1)$, where $x=w_2$ for some $w\in S_v$.}

\smallskip
\noindent The vertices of $T_x$ are $x,u_1,a_1,b_1$, and its edges are
$xu_1,xa_1,a_1b_1$.
If $M\cap\{u_1,b_1\}=\emptyset$, then
$M\in\mathcal{A}_v(T-\{u_1,b_1\})$.
Otherwise, neither $x$ nor $a_1$ belongs to $M$, since each is within distance
$2$ of both $u_1$ and $b_1$.
Consequently, $M\cap V(T_x)$ is exactly one of the three sets of the family
$\mathcal{S}=\{\{u_1\},\{b_1\},\{u_1,b_1\}\}$.
Since the restriction of $M$ to $T\setminus T_x$ is a multipacking,
$M\in\mathcal{S}\oplus\mathcal{A}_v(T\setminus T_x)$ for the family
$\mathcal{S}$ used in this case.

\medskip
\noindent\textit{\textbf{Case 4:} $T_x\cong H(0,2)$, where $x=w_2$ for some $w\in S_v$.}

\smallskip
\noindent The two paths of length $2$ are $xa_1b_1$ and $xa_2b_2$.
The complete family of multipackings of $T_x$ is $\mathcal{M}(T_x)=\{\emptyset,\{x\},\{a_1\},\{a_2\},\{b_1\},\{b_2\},\{a_1,b_2\},\{b_1,a_2\},\{b_1,b_2\}\}$. If $M\cap\{a_1,b_1,b_2\}=\emptyset$, then $M\cap V(T_x)$ is one of
$\emptyset,\{x\},\{a_2\}$, and
$M\in\mathcal{A}_v(T-\{a_1,b_1,b_2\})$.
In every other case, $M\cap V(T_x)$ is exactly one of the six sets of the family $\mathcal{S}=\{\{b_2\},\{a_1\},\{a_1,b_2\},\{b_1\},\{b_1,a_2\},\{b_1,b_2\}\}$. Thus
$M\in\mathcal{S}\oplus\mathcal{A}_v(T\setminus T_x)$ for the family
$\mathcal{S}$ used in this case.

\medskip
\noindent\textit{\textbf{Case 5:} None of the preceding cases applies.}

\smallskip
\noindent Choose any $w\in S_v$ and let $x=w_2$. We always have $M\in\mathcal{A}_v(T\setminus T_x)\oplus\mathcal{M}(T_x)$. 

\medskip

This proves $\mathcal{M}(T)\subseteq\mathcal{A}_v(T)$ in every case.
We do not require the converse inclusion: combinations that violate a
constraint across the two subtrees will be rejected when they are tested
in the original graph.

\medskip
We now bound the size of the generated family. Let $\beta=1.5283$, and
let $f(n)$ be the maximum number of sets returned by the algorithm on a
rooted tree with at most $n$ vertices, allowing the empty instance.
We use a simultaneous induction over all recursive cases to show that
$f(n)\leq C\beta^n$ for a fixed constant $C$.
For a tree of height at most $1$, the output has $n+1$ sets; hence the
base cases satisfy this bound for a suitable $C$.

In Case~1, at least three vertices are deleted when $w$ is selected,
and one vertex is deleted when $w$ is not selected. Thus the contribution
of this case is bounded by
$f(n-3)+f(n-1)$.
The same bound holds in Case~2 because $|V(H(0,1))|=3$.
In Case~3, the first recursive instance deletes two vertices; the other
three local choices each delete all four vertices of $T_x$.
The corresponding bound is $f(n-2)+3f(n-4)$.
In Case~4, the first recursive instance deletes three vertices, while
each of the other six local choices deletes all five vertices of $T_x$.
The corresponding bound is $f(n-3)+6f(n-5)$.
Direct calculation gives
$$
\beta^3>\beta^2+1,\qquad
\beta^4>\beta^2+3,\qquad
\beta^5>\beta^2+6.
$$
The smallest value of $\beta$ satisfying all the above inequalities is
the unique real root greater than $1$ of $z^5-z^2-6=0$, namely
$1.528219\ldots\leq 1.5283$. 
After substituting the inductive bounds, these inequalities imply an
upper bound of $C\beta^n$ in each of Cases~1--4.

Finally, in Case~5, let $m=|V(T_x)|$.
By Corollary~\ref{cor:H15283},
$|\mathcal{M}(T_x)|\leq\beta^m$.
Since the two factors use disjoint vertex sets, the output size is at most
$$
f(n)\leq f(n-m)\,|\mathcal{M}(T_x)|
\leq C\beta^{n-m}\beta^m=C\beta^n.
$$
Thus the same constant $C$ works throughout the induction, and
$|\mathcal{A}_v(T)|=O(1.5283^n)$.

Hence, the Algorithm \ref{alg:multipackingsuperset} computes a set $\mathcal{A}_v(T)$ of size $O(1.5283^n)$ that contains all  multipackings of $T$, and it has running time $O^*(1.5283^n)$.
\end{proof}

\multipackingtime*


\begin{proof} Let $G$ be a connected graph and $T$ be its spanning tree. Let $v$ be a vertex of $G$. Consider $T$ as a rooted tree with root $v$. Since $d_T(u,v)\geq d_G(u,v)$ for any $u,v\in V(G)$, $\mathcal{M}(G)\subseteq \mathcal{M}(T)$. From Lemma \ref{lem:15283setbound}, $\mathcal{A}_v(T)$ can be computed in time $O^*(1.5283^n)$ and $|\mathcal{A}_v(T)|=O(1.5283^n)$. Now $\mathcal{M}(G)\subseteq \mathcal{A}_v(T)$, since $\mathcal{M}(T)\subseteq \mathcal{A}_v(T)$. For any set $M\in \mathcal{A}_v(T)$, we can check whether $M$ is a multipacking of $G$ in $O(n^2)$ time. Therefore, we can search a maximum multipacking of $G$ in the family of sets $\mathcal{A}_v(T)$ in time $O^*(1.5283^n)$. Note that, $O^*(1.5283^n)=O(1.53^n)$.\footnote{For sufficiently large $n$, we have $1.5283^nn^{O(1)}< 1.53^n$. Hence $O^*(1.5283^n)$ can be replaced by $O(1.53^n)$.} If $G$ is not a connected graph, then a maximum multipacking can be computed independently for each connected component using the same approach.
\end{proof}

The exact exponential-time algorithm computes a maximum multipacking in an $n$-vertex graph in $O(1.53^n)$ time using a naive approach. We emphasize that this algorithm is not based on, and is not similar to, the known polynomial-time algorithm for trees~\cite{brewster2019broadcast}; here the spanning tree is used only to generate a superset of the set of possible solutions, since every multipacking of $G$ is also a multipacking of any spanning tree of $G$. The algorithm can be refined by exploiting structural properties of the rooted spanning tree. In particular, the current version considers only subtrees of height at most $2$; extending it to handle subtrees of greater height so that the possible multipackings within them can be computed using their structure could further improve the running time for finding a maximum multipacking of the whole graph. Note that in our algorithm the root of the spanning tree is fixed at each step. Another possible improvement is to update the root of the spanning tree at each step.


\begin{lemma}\label{lem:pathmultiapcking}
    If $P_n$ is a path on $n$ vertices, then the number of  multipackings of $P_n$ is $O(1.47^n) $.
\end{lemma}

\begin{proof} Let $P_n=\{v_1,v_2,\dots, v_n\}$ and $P_i$ be the induced subgraph of $P_n$ with vertex set $\{v_1,v_2,\dots, v_i\}$, where $i\leq n$. Suppose $M$ is a multipacking of $P_n$. If $v_n\in M$, then $v_{n-1},v_{n-2}\notin M$ since $|N_1[v_{n-1}]\cap M|\leq 1$.   From this observation, we can say that the number of multipackings of $P_n$ that include $v_n$ is $|\mathcal{M}(P_{n-3})|$ and the number of multipackings of $P_n$ that do not include $v_n$ is $|\mathcal{M}(P_{n-1})|$. Therefore, we have the recursive relation $|\mathcal{M}(P_n)|=|\mathcal{M}(P_{n-1})|+|\mathcal{M}(P_{n-3})|$ with initial conditions $|\mathcal{M}(P_1)|=2$, $|\mathcal{M}(P_2)|=3$, and $|\mathcal{M}(P_3)|=4$. Solving this recurrence yields $|\mathcal{M}(P_n)|\leq 2\cdot (1.46557...)^{n-1}$.
\end{proof}

  Lemma \ref{lem:pathmultiapcking} yields the following.



 \begin{theorem}
Let $G$ be an $n$-vertex graph, and suppose that a Hamiltonian path\footnote{A path $P$ in a graph $G$ is called a \emph{Hamiltonian path} if it contains every vertex of $G$ exactly once.} $P$ of $G$ is given as part of the input. Then a maximum multipacking of $G$ can be computed in time $O(1.47^n)$.
\end{theorem}



\begin{proposition}
    If $P_n$ is a path on $n$ vertices, then the number of maximal multipackings of $P_n$ is  $O(1.33^n) $.
\end{proposition}

\begin{proof}
Let $P_n=\{v_1,v_2,\dots, v_n\}$ and $P_i$ be the induced subgraph of $P_n$ with vertex set $\{v_1,v_2,\dots, v_i\}$, where $i\leq n$. Let $\mathcal{M}'(P_i)$ denote the set of all maximal multipackings in $P_i$. Suppose $M\in \mathcal{M}'(P_n)$. Then only one among  $v_n,v_{n-1},v_{n-2}$ is in $M$, since $|N_1[v_{n-1}]\cap M|\leq 1$.   The number of maximal multipackings of $P_n$ that include $v_n$ is $|\mathcal{M}'(P_{n-3})|$, since these multipackings cannot contain $v_{n-1}$ or 
$v_{n-2}$. The number of maximal multipackings of $P_n$ that include $v_{n-1}$ is $|\mathcal{M}'(P_{n-4})|$, since these multipackings cannot contain $v_n,v_{n-2}$ or 
$v_{n-3}$. Similarly, the number of maximal multipackings of $P_n$ that include $v_{n-2}$ is $|\mathcal{M}'(P_{n-5})|$, since these multipackings cannot contain $v_n,v_{n-1},v_{n-3}$ or 
$v_{n-4}$. Therefore, we have the recursive relation $|\mathcal{M}'(P_n)|=|\mathcal{M}'(P_{n-3})|+|\mathcal{M}'(P_{n-4})|+|\mathcal{M}'(P_{n-5})|$ with initial conditions $|\mathcal{M}'(P_1)|=1$, $|\mathcal{M}'(P_2)|=2$, and $|\mathcal{M}'(P_3)|=3$. Solving this recurrence yields $|\mathcal{M}'(P_n)|\leq  (1.3247...)^{n}$.
\end{proof}

\section{Conclusion}\label{sec:conclusion}

We conclude with some remarks  and open problems concerning the \textsc{Multipacking} problem.

 It is known that planar graphs are SEG graphs\footnote{ A graph $G$ is a \textit{SEG graph} (or \textit{Segment Graph}) if and only if there exists  a family of straight-line segments in a plane such that the graph has a vertex for each straight-line segment and an edge for each intersecting pair of straight-line segments.}~\cite{chalopin2009every}. In 1998, Kratochv{\'\i}l and Kub{\v{e}}na~\cite{kratochvil1998intersection} posed the question of whether the complement of any planar graph (i.e., a co-planar graph) is also a SEG graph. As far as we know, this is still an open question. A positive answer to this question would imply that the \textsc{Multipacking} problem is \textsc{NP}-complete  for SEG graphs, since in that case the induced subgraph $G'[u_1^1, u_2^1, \dots, u_n^1]$ in the proof of Lemma~\ref{thm:multipacking_CONV_NPc} would be a SEG graph, which in turn would imply that $G'$ itself is a SEG graph. However, the complexity of the problem on planar graphs still remains unresolved.

Our results in this paper naturally lead to several open questions that merit further investigation:
\begin{enumerate}

    \item What is the complexity status of the problem on planar graphs or bounded degree graphs? Does an FPT algorithm for planar graphs exist when parameterized by solution size?

    \item Does there exist an FPT algorithm for general graphs when parameterized by treewidth?
    \item Can one design a subexponential-time algorithm (that is, an algorithm running in time $2^{o(n)}$) for general graphs?

    \item While a $(2+o(1))$-factor approximation algorithm is known~\cite{beaudou2019broadcast}, can the approximation factor be improved? Moreover, does the problem admit a PTAS?

    
\end{enumerate}

Addressing these questions would provide a more complete understanding of the computational landscape of the \textsc{Multipacking} problem and could lead to important algorithmic developments.

\backmatter

\section*{Statements and Declarations}

\noindent\textbf{Funding.}
Lokshtanov was supported by NSF Grant  CCF-2505099.

\medskip

\noindent\textbf{Competing Interests.}
The authors have no relevant financial or non-financial interests to disclose.

\medskip

\noindent\textbf{Data Availability.}
No datasets were generated or analyzed during the current study.

\bibliography{sn-bibliography}

@book{erwin2001cost,
  title={Cost domination in graphs},
  author={Erwin, David J},
  year={2001},
  publisher={Western Michigan University},
  address = {Kalamazoo, MI}
}

@article{erwin2004dominating,
  title={Dominating broadcasts in graphs},
  author={Erwin, David J},
  journal={Bull. Inst. Combin. Appl},
  volume={42},
  number={89},
  pages={105},
  year={2004}
}

@article{beaudou2019multipacking,
  title={On the multipacking number of grid graphs},
  author={Beaudou, Laurent and Brewster, Richard C},
  journal={Discrete Mathematics \& Theoretical Computer Science},
  volume={21},
  number={Graph Theory},
  year={2019},
  publisher={Episciences. org}
}

@article{brewster2013new,
  title={New bounds for the broadcast domination number of a graph},
  author={Brewster, Richard C and Mynhardt, Christina M and Teshima, Laura E},
  journal={Central European Journal of Mathematics},
  volume={11},
  pages={1334--1343},
  year={2013},
  publisher={Springer}
}

@article{brewster2019broadcast,
  title={Broadcast domination and multipacking in strongly chordal graphs},
  author={Brewster, Richard C and MacGillivray, Gary and Yang, Feiran},
  journal={Discrete Applied Mathematics},
  volume={261},
  pages={108--118},
  year={2019},
  publisher={Elsevier}
}

@article{dunbar2006broadcasts,
  title={Broadcasts in graphs},
  author={Dunbar, Jean E and Erwin, David J and Haynes, Teresa W and Hedetniemi, Sandra M and Hedetniemi, Stephen T},
  journal={Discrete Applied Mathematics},
  volume={154},
  number={1},
  pages={59--75},
  year={2006},
  publisher={Elsevier}
}

@article{heggernes2006optimal,
  title={Optimal broadcast domination in polynomial time},
  author={Heggernes, Pinar and Lokshtanov, Daniel},
  journal={Discrete Mathematics},
  volume={306},
  number={24},
  pages={3267--3280},
  year={2006},
  publisher={Elsevier}
}

@article{foucaud2021complexity,
  title={On the complexity of broadcast domination and multipacking in digraphs},
  author={Foucaud, Florent and Gras, Benjamin and Perez, Anthony and Sikora, Florian},
  journal={Algorithmica},
  volume={83},
  number={9},
  pages={2651--2677},
  year={2021},
  publisher={Springer}
}

@article{meir1975relations,
  title={Relations between packing and covering numbers of a tree},
  author={Meir, A and Moon, John},
  journal={Pacific Journal of Mathematics},
  volume={61},
  number={1},
  pages={225--233},
  year={1975},
  publisher={Mathematical Sciences Publishers}
}

@mastersthesis{teshima2012broadcasts,
  title={Broadcasts and multipackings in graphs},
  author={Teshima, Laura E},
  school = {University of Victoria},
  year={2012}
}

@mastersthesis{yang2015new,
  title={New results on broadcast domination and multipacking},
  author={Yang, Feiran},
  school = {University of Victoria},  
  year={2015}
}

@phdthesis{yang2019limited,
  title={Limited broadcast domination},
  author={Yang, Feiran},
  school = {University of Victoria},  
  year={2019}
}

@article{beaudou2019broadcast,
  title={Broadcast domination and multipacking: bounds and the integrality gap},
  author={Beaudou, Laurent and Brewster, Richard C and Foucaud, Florent},
  journal={The Australasian Journal of Combinatorics},
  volume={74},
  number={1},
  pages={86--97},
  year={2019}
}

@inproceedings{das2025hardnessMultipackingESA,
  title={On the {C}omplexity of {M}ultipacking},
  author={Das, Sandip and Islam, Sk Samim and Lokshtanov, Daniel},
  booktitle={European Symposium on Algorithms (ESA)},
  volume    = {388},
  pages     = {152:1--152:23},
  year={2026},
  organization={}
}

@book{cornuejols2001combinatorial,
  title={Combinatorial Optimization: Packing and Covering},
  author={Cornu{\'e}jols, G{\'e}rard},
  year={2001},
  publisher={SIAM},
  address={Philadelphia, PA}
}

@article{brinkmann2001hyperbolicity,
  title={On the hyperbolicity of chordal graphs},
  author={Brinkmann, Gunnar and Koolen, Jack H and Moulton, Vincent},
  journal={Annals of Combinatorics},
  volume={5},
  number={1},
  pages={61--69},
  year={2001},
  publisher={Springer}
}

@article{wu2011hyperbolicity,
  title={Hyperbolicity and chordality of a graph},
  author={Wu, Yaokun and Zhang, Chengpeng},
  journal={The Electronic Journal of Combinatorics},
  volume={18},
  number={1},
  pages={P43},
  year={2011}
}

@incollection{gromov1987hyperbolic,
  title={Hyperbolic groups},
  author={Gromov, Mikhael},
  booktitle={Essays in group theory},
  pages={75--263},
  year={1987},
  publisher={Springer},
  address   = {New York, NY}
}

@book{haynes2021structures,
  title={Structures of domination in graphs},
  author={Haynes, Teresa W and Hedetniemi, Stephen T and Henning, Michael A},
  volume={66},
  year={2021},
  publisher={Springer},
  address= {Cham}
}

@book{cygan2015parameterized,
  title={Parameterized algorithms},
  author={Cygan, Marek and Fomin, Fedor V and Kowalik, {\L}ukasz and Lokshtanov, Daniel and Marx, D{\'a}niel and Pilipczuk, Marcin and Pilipczuk, Micha{\l} and Saurabh, Saket},
  year={2015},
  publisher={Springer},
  address= {Cham}
}

@article{kratochvil1998intersection,
  title={On intersection representations of co-planar graphs},
  author={Kratochv{\'\i}l, Jan and Kub{\v{e}}na, Ale{\v{s}}},
  journal={Discrete Mathematics},
  volume={178},
  number={1-3},
  pages={251--255},
  year={1998},
  publisher={Elsevier}
}

@inproceedings{chalopin2009every,
  title={Every planar graph is the intersection graph of segments in the plane: extended abstract},
  author={Chalopin, J{\'e}r{\'e}mie and Gon{\c{c}}alves, Daniel},
  booktitle={STOC '09: Proceedings of the 41st annual ACM symposium on Theory of computing},
  pages={631--638},
  year={2009}
}

@article{garey1979computers,
  title={Computers and Intractability: A Guide to the theory of {NP}-completeness.},
  author={Michael R. Garey and David S. Johnson},
  journal={W.H. Freeman, San Francisco, 338 pp.},
  volume={},
  number={},
  pages={},
  year={1979},
  publisher={}
}

@article{erdos1960grafok,
  title={Gr{\'a}fok el{\H{o}}{\'\i}rt foksz{\'a}m{\'u} pontokkal},
  author={Erd{\H{o}}s, Paul and Gallai, Tibor},
  journal={Matematikai Lapok},
  volume={11},
  pages={264--274},
  year={1960}
}

@article{hakimi1962realizability,
  title={On realizability of a set of integers as degrees of the vertices of a linear graph. {I}},
  author={Hakimi, S Louis},
  journal={Journal of the Society for Industrial and Applied Mathematics (SIAM)},
  volume={10},
  number={3},
  pages={496--506},
  year={1962},
  publisher={SIAM}
}

@article{havel1955remark,
  title={A remark on the existence of finite graphs},
  author={Havel, V{\'a}clav},
  journal={Casopis Pest. Mat.},
  volume={80},
  pages={477--480},
  year={1955}
}

@inproceedings{das2023relation,
  title={Relation Between Broadcast Domination and Multipacking Numbers on Chordal Graphs},
  author={Das, Sandip and Foucaud, Florent and Islam, Sk Samim and Mukherjee, Joydeep},
  booktitle={Conference on Algorithms and Discrete Applied Mathematics ({CALDAM})},
  pages={297--308},
  year={2023},
  organization={Springer}
}

@inproceedings{das2025multipacking,
  title={Multipacking and broadcast domination on cactus graphs and its impact on hyperbolic graphs},
  author={Das, Sandip and Islam, Sk Samim},
  booktitle={International Conference and Workshops on Algorithms and Computation (WALCOM)},
  pages={111--126},
  year={2025},
  organization={Springer}
}

@inproceedings{das2025geometrymultipackingCALDAM,
  title={Multipacking in the {E}uclidean Metric Space},
  author={Das, Arun Kumar and Das, Sandip and Islam, Sk Samim and Mitra, Ritam Manna and Roy, Bodhayan},
  booktitle={Conference on Algorithms and Discrete Applied Mathematics (CALDAM)},
  pages={109--120},
  year={2025},
  organization={Springer}
}

@article{rajendraprasad2025multipacking,
  title={Multipacking in Hypercubes},
  author={Rajendraprasad, Deepak and Sani, Varun and Sasidharan, Birenjith and Sen, Jishnu},
  journal={arXiv preprint arXiv:2507.01565},
  year={2025}
}

\end{document}